\documentclass[lettersize,onecolumn]{IEEEtran}
\usepackage{booktabs}
\usepackage{amsmath}
\usepackage{amssymb}
\usepackage{caption}
\usepackage{algorithmic}
\usepackage{todonotes}
\usepackage{amsfonts}
\usepackage{cite}
\usepackage{graphicx}
\usepackage{array}
\usepackage[shortlabels]{enumitem}
\usepackage{soul}
\newtheorem{definition}{Definition}[section]
\newtheorem{theorem}[definition]{Theorem}
\newtheorem{corollary}[definition]{Corollary}
\newtheorem{lemma}[definition]{Lemma}

\newtheorem{example}[definition]{Example}
\newtheorem{proposition}[definition]{Proposition}
\newtheorem{remark}[definition]{Remark}

\begin{document}
\def\Out{{\rm Out}}
\def\In{{\rm In}}
\title{On Minimum Distances for Error Correction and Detection of Generalized Network Code}
%
%
%

\author{Yulin~CHEN
        
        Raymond~W.~Yeung,~\IEEEmembership{Fellow,~IEEE}
\thanks{Yulin Chen is with the Department
of Information Engineering, The Chinese University of Hong Kong, Hong Kong, SAR, China e-mail: (1155189560@link.cuhk.edu.hk)}
\thanks{Raymond\ W.\ Yeung is with the Department of Information Engineering and the Institute of Network Coding, The Chinese University of Hong Kong, Hong Kong, SAR, China (e-mail:whyeung@ie.cuhk.edu.hk)}
}
\maketitle

\begin{abstract}
It is well known that the minimum distance for linear network codes plays the same role as the minimum distance for classical error control codes. However, Yang {\em et al.} (2008) discovered that for nonlinear network codes, the minimum distance for error correction is not always the same as the minimum distance for error detection. In this paper, we define two distances for {\em joint} error correction and detection, one being a refinement of the other, and obtain their basic properties. Then we study the relations among these two distances and the two previously defined distances for error correction and error detection. In particular, it is proved that in the case of a linear network code, only one minimum distance is needed to completely characterize the various capabilities of the code for error correction and detection. At the end of the paper, we introduce the \emph{generalized network channel} as an abstraction of a network from the input-output point of view. This enables our results to be applied in a very general setting, which induces the rank metric code as a special case.
\end{abstract}


%
\IEEEpeerreviewmaketitle

\section{Introduction}
\IEEEPARstart{N}{etwork} error correction was first studied by Yeung and Cai in \cite{neterror,networkerror}. Subsequently, Zhang \cite{minidis} introduced an algebraic definition of minimum distance for network error correction, and Yang {\em et al.} \cite{WeightProperty} introduced a geometrical definition of minimum distance that was recently shown by Guang and Yeung \cite{GY23} to be equivalent to the one introduced in \cite{minidis}. This minimum distance fully characterizes the error correction capability of a {\em linear} network code. 


In addition to network error correction, network error detection was also considered in \cite{WeightProperty}. In particular, it was shown that for a {\em nonlinear} network code, two different distances are needed for characterizing its capabilities for error correction and error detection. This leads to the surprising discovery that for a nonlinear network code, unlike a classical block channel code, the number of correctable errors can be more than half of the number of detectable errors. Moreover, a framework that captures joint error correction and detection was introduced.

 In the above line of research, it is assumed that the network can be used once, and both the transmitter and the receiver have full knowledge of the network. 
 In another line of research started by Koetter and Kschischang \cite{subcode}, which was inspired by random linear network coding \cite{rlnc},
it is assumed that the network can be used multiple times, but the transmitter and the receiver
 do not need to have any knowledge on the network and its transfer characteristics except that the underlying network code
 is linear. This line of research is referred to as {\em noncoherent} (linear) network coding, in order to distinguish it from the original line of research on network coding.  In this paper, a network code is assumed to be coherent unless otherwise specified.
 

The current paper is a continuation of the studies in \cite{WeightProperty}. Our contributions are two-fold:
\begin{enumerate}
\item 
We formulate the transmission system studied in \cite{WeightProperty} as the {\em generalized network channel} and the {\em generalized network code}. We study the interplay between different error control functions (error correction, error detection, and joint error correction and detection) of a generalized network code by establishing fundamental bounds on the distances associated with these error control functions.
\item 
We formulate the notions of {\em error-linearity} for a generalized network channel and {\em linearity} for a generalized network code. In particular, all the following can be regarded as special cases of a linear generalized network code over an error-linear generalized network channel:
\begin{enumerate}
	\item	
	a block code in classical channel coding
	\item 
	a coherent linear network code
	\item 
	a rank metric code
	\item 
	a sum-rank metric code.
\end{enumerate}
We prove that for an error-linear generalized network channel, the distances for error correction, error detection, and joint error correction and detection coincide with each other, generalizing the corresponding result in \cite{WeightProperty}.
\end{enumerate}

The rest of the paper is organized as follows. In Section~\ref{sec7} we define the generalized network channel and the generalized network code. In Section~\ref{sec8}, four codes from the classical channel coding and network coding are casted into the framework introduced in Section~\ref{sec7}. In Section~\ref{sec4}, we obtain a lower bound on 
the minimum error correction distance in terms of the minimum error detection distance. Section~\ref{sec5} is a discussion on joint network error correction and detection that contains the main results of this paper. In this section, we introduce a distance and its refined version for joint error correction and detection and prove some of their basic properties. We also obtain various bounds between the distances for error correction, error detection, and joint error correction and detection. Section~\ref{sec6} uses the toy example in \cite{WeightProperty} to illustrate the results in the previous sections. The paper is concluded in Section~\ref{sec9}.

 \vspace{2mm}

\section{Generalized Network Channel and Code}\label{sec7}

In \cite{WeightProperty},  error correction and error detection were studied for a general transmission system which includes network coding as a special case. It was shown that for a {\em nonlinear} network code, two different distances are needed for characterizing its capabilities for error correction and error detection. An example of a nonlinear network code was given to illustrate that these two distances can in fact be different. Specifically, this nonlinear network code can correct one error and detect one error. This example is very intriguing because in classical channel coding, the number of errors that can be corrected is at most half of the number of errors that can be detected.

On the other hand, it was proved in \cite{WeightProperty} that for a coherent linear network code, the distances for error correction and error detection coincide with each other, so that we only need one distance to characterize the capabilities of such a code for error correction and error detection.

In Section~\ref{sec5} (Theorem~\ref{d_2theo}), we will establish the same result but in a much more general context. To start with, we formulate in this section the transmission system studied in \cite{WeightProperty} as the generalized network channel and the generalized network code. In the next section, we will show that our formulation includes a block code in classical channel coding, a coherent linear network code, a rank metric code, and a sum-rank metric code as special cases.

Since the network topology and coding coefficients can be represented only by the transfer function $\tilde{F}$, 
we are motivated to introduce the {\em generalized network channel} as an abstraction of a network from the input-output point of view, with the aim to establish a more general coding scheme that focuses on whether the channel is ``linear''. 

 To formulate the generalized network channel, we let ${\cal C}$ and ${\cal E}$ be the set of codewords and the set of errors, respectively. The set ${\cal C}$ and the error set ${\cal E}$ can be arbitrary. The error set ${\cal E}$ is equipped with a weight measure $w: {\cal E} \rightarrow \mathbb{N} \cup \{ 0 \}$ that satisfies the following property:
\begin{enumerate}
\item
$w(z) \ge 0$ for all $z \in {\cal E}$, with equality if and only if $z = 0$. The error $z=0$ is called the {\em zero error}.
\end{enumerate}


Let ${\cal Y}$ be the output set of the generalized network channel, which can be arbitrary.
The generalized network channel is specified by a transfer function
${F}: {\cal C} \times {\cal E} \rightarrow {\cal Y}$. This means that when a codeword $x \in {\cal C}$ is transmitted and an
error $z \in {{\cal E}}$ occurs, the channel outputs a received word $y$, where $y \in {\cal Y}$. In order to be able to recover the transmitted codeword from the received word in the absence of error, i.e., when $z=0$, we assume that for $x,x' \in {\cal C}$ with $x \ne x'$,
\begin{equation}\label{basicequation}
    {F}(x,0) \ne {F}(x',0).
\end{equation}

When there is no ambiguity, we may also refer to the generalized network channel as ${F}$.

A generalized network channel is called {\em error-linear} if 
the following properties are satisfied:
\begin{enumerate}
\setcounter{enumi}{1}
\item ${\cal E}$ is a group, with the binary operator denoted by `$\circ$', and the identity element, namely the  zero error, denoted by 0;
\item 
${\cal Y}$ is an abelian group, with the commutative binary operator denoted by `$\oplus$', and the identity element, namely the \emph{zero received word}, denoted by $0$;\footnote{The identity elements of $\mathcal{E}$ and $\mathcal{Y}$ are both denoted by $0$, which should incur no ambiguity.}
\item \label{qwp9uhg}
$ F$ is a ``linear function'' in the sense that ${F}(x,z) = {f}(x) \oplus {h}(z)$, where ${f}: {\cal C} \rightarrow {\cal Y}$, ${h}: {\cal E} \rightarrow {\cal Y}$,
and 
\[
{h}(z \circ z') = {h}(z) \oplus {h}(z') ,
\]
for all $z, z' \in {\cal E}$, i.e., ${h}$ is a group homomorphism; 

\item 
(Triangular inequality) $w(z \circ z') \le w(z) + w(z')$ for all $z, z' \in {\cal E}$;
\item 
$w(z^{-1}) = w(z)$ for all $z \in {\cal E}$, where $z^{-1}$ denotes the inverse of $z$ in ${\mathcal{E}}$;
\item 
For any $z \in {\cal E}$ and $c_1, c_2 \ge 0$ such that $c_1+c_2 = w(z)$, there exist $z_1, z_2 \in {\cal E}$ such that $w(z_1) = c_1$ and $w(z_2) = c_2$,  and $z_1\circ z_2=z$.
\end{enumerate}
Furthermore, an error-linear generalized network channel is called {\em linear} if the following additional properties are satisfied:
\begin{enumerate}
\setcounter{enumi}{7}
\item 
 ${\cal E}$ and ${\cal Y}$ are vector spaces;
\item 
${f}$ is linear.
\end{enumerate}
If a generalized network channel is not linear, then it is said to be \emph{nonlinear.}

Note that in Property 4), since ${h}$ is a group homomorphism, for all $z\in\mathcal{E}$, we have \[{h}(z)={h}( 0\circ z)={h}(0)\oplus{h}(z),\] so ${h}(0)=0.$ On the other hand, for any $z\in \mathcal{E}$, since \[0={h}(0)={h}(z\circ z^{-1})={h}(z)\oplus {h}(z^{-1}),\] we have \[{h}(z^{-1})=({h}(z))^{-1}.\]
These properties will be useful in our subsequent proofs.

A generalized network code consists of the codeword set ${\cal C}$ and the generalized network channel ${F}$, and is denoted as $(\mathcal
{C},F)$. Furthermore, if $F$ is linear and  
\begin{enumerate}
\setcounter{enumi}{9}
\item $\mathcal{C}$ is a vector space,
\end{enumerate}
then the generalized network code is called linear.

For a linear generalized network channel, Properties~8) and 9) are satisfied. Then the binary operators ``$\circ$'' and ``$\oplus$'' for $\mathcal{E}$ and $\mathcal{Y}$, respectively are the vector additions in the corresponding vector spaces. Furthermore, 
\[
{F}(x,z) = x \cdot A + z \cdot B,
\]
where `$\cdot$' is matrix multiplication and `$+$' is matrix addition, and $A$ and $B$ are matrices of suitable dimensions.

\section{Classical Block Codes, Coherent Network Coding, and Noncoherent Network Coding}\label{sec8}
 In the definition of the generalized network channel in Section \ref{sec7}, the error set ${\mathcal{E}}$ is equipped with a weight measure $w$ that satisfies Property~1). A generalized network code consists of the codeword set ${\cal C}$ and the generalized network channel $F$. 
If the generalized network code is error-linear, then $w$ further satisfies Properties 5) to 7).
Obviously, the Hamming weight satisfies all these properties. 

In the literature of classical coding theory and network coding theory,
there are a lot of weight measures in use, for example, the Hamming weight, the rank weight \cite{koetter_1} and the sum-rank weight \cite{elgamalsum},\cite{sum-rank}. In this section, we demonstrate the usefulness of the generalized framework introduced in Section~\ref{sec7} by showing that the block code through a classical point-to-point channel equipped with the Hamming weight $w_H$, the rank metric code through a non-coherent linear network channel equipped with the rank weight rk, and the sum-rank metric code through a non-coherent linear network channel equipped with the sum-rank metric weight sr, are indeed linear generalized network codes. Besides, we also show that a network code through a coherent linear network channel equipped with the Hamming weight is indeed a linear generalized network code, otherwise, it is a nonlinear generalized network code.

Throughout this paper, we use $\mathbb{F}_q$ to denote the Galois field of order $q$ with $q = p^k$, where $p$ is a prime and $k$ is a positive integer. The field  $\mathbb{F}_{q}$ is called the extension field of $\mathbb{F}_p$ with extension degree $k$. Let $\mathbb{F}_q^m$ denote the vector space of row vectors of length $m$ over the base field $\mathbb{F}_q$, and $\mathbb{F}_q^{m\times n}$ denote the vector space of all $m \times n$ matrices over $\mathbb{F}_q$.\\

\subsection{Block Code over the Classical Point-to-Point Channel}
 Consider a block code of length $n$ over  $\mathbb{F}_{q}$ in classical channel coding. When a codeword $x\in\mathbb{F}_q^n$ is transmitted through a classical point-to-point channel $F$, an error vector $z\in\mathbb{F}_q^n$ may occur, and the vector $y=x+z$ is received. 
This can be cast into the framework of generalized network code as follows:
 \begin{enumerate}[i)]
     \item The codeword set $\mathcal{C}$ is a subset of $\mathbb{F}_q^n$. The error set $\mathcal{E}$ is $\mathbb{F}_q^n$ and the received word set $\mathcal{Y}$ is $\mathbb{F}_q^n$.
     \item The weight measure $w$ is the Hamming weight $w_H$.
     \item The transfer function that represents the point-to-point channel is $F(x,z)=x+z.$
 \end{enumerate}

 In classical channel coding, the distance $D$ that measures the distance between two vectors is the Hamming distance, given by $D(x_1,x_2)=w_H(x_1-x_2)$. In this case, $D(x_1,x_2)$ is a function of $x_1-x_2$. However, as we will see in Section \ref{sec4}, in the case of the generalized network code, since $\mathcal{C}$ is not necessarily a group (and so $x_1-x_2$ is undefined), the distance between $x_1$ and $x_2$ will be obtained through the intersection of Hamming balls.
 
Now we show that a classical point-to-point channel is a linear generalized network channel by verifying Properties 1) to 9) in Section \ref{sec7}:
\begin{enumerate}
    \item  $w_H(z)\ge 0$ with equality if and only if $z=0$.
    \item $\mathcal{E}=\mathbb{F}_{q}^n$ is a group with the binary operator `$\circ$' being the addition in $\mathbb{F}_q^n$, and the zero error is the all-zero vector $0$.
    \item $\mathcal{Y}=\mathbb{F}_q^n$ is an abelian group with the binary operator `$\oplus$' being the addition in $\mathbb{F}_q^n$. The zero received word is the all-zero vector 0.
    \item ${F}(x,z)=x+z={f}(x)+{h}(z)$,where $ {f}:{\mathcal{C}}\rightarrow{\mathcal{Y}}$ and ${h}:{\mathcal{E}}\rightarrow{\mathcal{Y}}$ are identity functions. Evidently, $h$ is a homomorphism because $h(z+z')=z+z'=h(z)+h(z')$.
    \item $w_H(z+z')\leq w_H(z)+w_H(z')$.
    \item  $w_H(-z)=w_H(z)$.
    \item Evidently, for any $z \in {\cal E}$ and $c_1, c_2 \ge 0$ such that $c_1+c_2 = w_H(z)$, there exist $z_1, z_2 \in {\cal E}$ such that $w_H(z_1) = c_1$ and $w_H(z_2) = c_2$,  and $z_1+z_2=z$.
    \end{enumerate}
    
    Therefore, a classical point-to-point channel is an error-linear generalized network channel. Moreover:
\begin{enumerate}
    \setcounter{enumi}{7}
    \item By definition, $\cal E$ and $\cal Y$ are vector spaces.
    \item $f:\mathcal{C}\rightarrow\mathbb{F}_q^n$ is the identity function in $\mathbb{F}_q^n$, so $f$ is linear.
\end{enumerate}
Therefore, a classical point-to-point channel under the Hamming weight is a linear generalized network channel. 
Furthermore, if
\begin{enumerate}
    \setcounter{enumi}{9}
    \item $\mathcal{C}$ is a vector space,
\end{enumerate}
then the block code is linear.

\subsection{Network Code over a Coherent Network Channel}
We consider the following model for network coding in the presence of errors. Let $\mathcal{G}=(V,E)$ represent an acyclic graph with $V$ being the set of nodes and $E$ being the set of edges. The edges in $E$ are labeled by $e_1,e_2,\cdots, e_{|E|}$. The codeword set ${\mathcal{C}}$ is a set of vectors ${\mathcal{C}}=\{x_1,x_2,\cdots x_{|{\mathcal{C}|}}\}$ with $x_i\in \mathbb{F}_q^{|\Out(s)|}$, where $\Out(s)$ is the set of output edges from the source node~$s$. To avoid triviality, we assume that $|{C}|\geq 2$.  The received word set $\mathcal{Y} $ is $\mathbb{F}_q^{|\In(t)|}$, where $\In(t)$ is the set of incoming edges of the sink node $t$. Let $T$ be the set of all sink nodes in the network. 

The set of error vectors, denoted by ${\mathcal{E}}$, is equal to ${F}_q^{|E|}$. For an error vector $z\in{\mathcal{E}}$, if the $i^{th}$ symbol of $z$ is $m\in\mathbb{F}_q$, it means that an error with value $m$ occurs on the $i^{th}$ edge of the network. When $m=0$, we say that a zero error occurs, or no error occurs. Note that an error can occur at any edge of the network.

The network channel is specified by a function ${F_t}: {\mathcal{C}}\times\mathcal{E}\rightarrow\mathbb{F}_q^{ |\In(t)|}$. When a codeword $x\in\mathcal{C}$ is transmitted and an error vector $z\in\mathcal{E}$ occurs in the network, the vector $F_t(x,z)$ is received at the sink node $t$.\footnote{Note that $F_t(x,z)$ depends on the network topology and the coding coefficients of the network code at the intermediate nodes. These details can be omitted for our discussion here.} If $F_t$ is known to the sink node for all $t\in T$, then we say that the network channel is a \emph{coherent} network channel. Otherwise, it is a \emph{non-coherent} network channel. Also, we assume that for $x_1,x_2\in{\mathcal{C}}$ with $x_1\neq x_2$, 
 \begin{equation}
     {F_t}(x_1,0)\neq {F_t}(x_2,0).
 \end{equation}
This means that in the absence of errors, the codeword transmitted by the source node $s$ can be correctly decoded at every sink node $t$. Since in this work we only focus on a single sink node, we will omit the subscript $t$ in the sequel.

In this section, we consider a network code over a coherent network channel, which can be cast into the framework of generalized network code as follows:
\begin{enumerate}[i)]
    \item The codeword set $\mathcal{C}$ is a subset of $\mathbb{F}_{q}^m$, where $m=|\Out(s)|$. The error set $\mathcal{E}$ is $\mathbb{F}_{q}^{|E|}$, where $|E|$ is the number of edges in the network, and the received word set $\mathcal{Y}$ is $\mathbb{F}_{q}^n$, where $n=|\In(t)|$.
     \item The weight measure $w$ is the Hamming weight $w_H$.
     \item We say the network channel is linear if  $F(x,z)$ is linear, i.e., $F(x,z)=x\cdot F_{s,t}+z\cdot H_t$, where $F_{s,t}$ and $H_t$ are the network transfer matrices for codeword transmission and error transmission, respectively\cite{book}. Otherwise, the network channel is said to be nonlinear. Note that if $F(x,z)=f(x)+h(z)$, with $f(x)$ being a nonlinear function on $\mathcal{C}$ and $h(z)$ being a linear function on $\mathcal{E}$, then $F(x,z)$ is a nonlinear network channel by definition. 
\end{enumerate}

Now we show that a coherent linear network channel is a linear generalized network channel by verifying Properties 1) to 9) in Section \ref{sec7}:
\begin{enumerate}
    \item $w_H(z)\ge 0$ with equality if and only if $z=0$.
    \item $\mathcal{E}=\mathbb{F}_{q}^{|E|}$ is a group with the binary operator `$\circ$' being the addition in $\mathbb{F}_{q}^{|E|}$, and the zero error is the all-zero vector $0$.
    \item $\mathcal{Y}=\mathbb{F}_{q}^n$ is an abelian group with the binary operator `$\oplus$' being the addition in $\mathbb{F}_{q}^n$. The zero received word is the all-zero vector 0.
    \item ${F}(x,z)={f}(x)+{h}(z)$, where $ {f(x)}=x\cdot F_{s,t}$ and ${h(z)=z\cdot H_t}$ are linear functions. 
    \item  $w_H(z+z')\leq w_H(z)+w_H(z')$.
    \item  $w_H(-z)=w_H(z)$.
    \item Evidently, for any $z \in {\cal E}$ and $c_1, c_2 \ge 0$ such that $c_1+c_2 = w_H(z)$, there exist $z_1, z_2 \in {\cal E}$ such that $w_H(z_1) = c_1$ and $w_H(z_2) = c_2$,  and $z_1+z_2=z$.
    \end{enumerate}
    
    Therefore, a coherent linear network channel under the Hamming weight is an error-linear generalized network channel. Moreover:
\begin{enumerate}
    \setcounter{enumi}{7}
    \item By definition, $\cal E$ and $\cal Y$ are vector spaces.
    \item  $f$ is linear.
\end{enumerate}
Therefore, a coherent linear network channel under the Hamming weight is a linear generalized network channel. Furthermore, if
\begin{enumerate}
    \setcounter{enumi}{9}
    \item $\mathcal{C}$ is a vector space,
\end{enumerate}
then the network code is linear, and $(C,F)$ forms a linear generalized network code.


\subsection{Rank Metric Code over a Non-coherent Linear Network Channel}
The rank metric code, introduced in \cite{coherent}, is a subspace code equipped with a metric induced by the rank weight that measures the distance between two subspaces. The rank metric code can be used over a random linear network channel\cite{silva2008rank}, which is a non-coherent linear network channel. A rank metric code over a non-coherent linear network channel can be cast into the framework of generalized network code as follows: 

\begin{enumerate}[i)]
    \item The codeword set ${\mathcal{C}}$ is the matrix space $\mathbb{F}_q^{m\times k}$, and the set of errors ${\mathcal{E}}$ is $\mathbb{F}_q^{m\times u}$, with $u$ being an integer representing the length of the error vector. The received word set ${\mathcal{Y}}$ is the matrix space $ \mathbb{F}_q^{m\times n}$. 

    \item The weight measure is the rank weight rk, where $\mbox{rk}(z)$ is the column rank of the error matrix $z$.

    \item The transfer function which represents a non-coherent linear network channel is ${F}(x,z)=x\cdot F_{s,t}+z\cdot H_t$, where $F_{s,t}$ and $H_t$ are transfer matrices of the network for the codeword and the error, respectively. 
\end{enumerate}

To show that the non-coherent linear network channel associated with the rank metric code is a linear generalized network channel, we need to verify Properties 1) to 9) in Section \ref{sec7}:

 \begin{enumerate}
     \item $\mbox{rk}(z)\geq 0$ for all $z\in{\mathcal{E}}$, with equality if and only if $z$ is the zero matrix.
 \end{enumerate}

\begin{enumerate}
\setcounter{enumi}{1}
      \item ${\mathcal{E}}$ is the matrix space $\mathbb{F}_q^{m\times u}$, which is a group with the binary operator `$\circ$' being matrix addition, and the zero error is the all-zero matrix $0$.
     \item ${\mathcal{Y}}$ is the matrix space $\mathbb{F}_q^{m\times n}$, which is an abelian group with the commutative binary operator `$\oplus$' being matrix addition. The zero received word is the all-zero matrix 0.
     \item ${F}(x,z)={f}(x)+{h}(z)$, where $ {f(x)}=x\cdot F_{s,t}$ and ${h(z)=z\cdot H_t}$ are linear functions. Here, `$+$' denotes addition in the matrix space $\mathbb{F}_q^{l\times n}$. 
     \item  $\mbox{rk}(z+z')\leq \mbox{rk}(z)+\mbox{rk}(z')$, where the `$+$' in $z+z'$ is addition in the matrix space $\mathbb{F}_q^{l\times u}$ and the `$+$' in $\mbox{rk}(z)+\mbox{rk}(z')$ is addition in the integer group.
     \item  $\mbox{rk}(-z)=\mbox{rk}(z)$, with $-z$ being the additive inverse of $z$, i.e., $z+(-z)=0$.
     \item We will prove in Appendix A that for every $z \in {\cal E}$ and $c_1, c_2 \ge 0$ such that $c_1+c_2 = \mbox{rk}(z)$ and $c_1+c_2\leq l$, we can always find two matrices $z_1$ and $z_2$ such that $\mbox{rk}(z_1) = c_1$, $\mbox{rk}(z_2) = c_2$, and $z_1+z_2=z$. 
 \end{enumerate}
Therefore, the non-coherent linear network channel under the rank weight is an error-linear generalized network channel. Moreover:
\begin{enumerate}
\setcounter{enumi}{7}
    \item By definition, ${\cal E}$ and ${\cal Y}$ are vector spaces.
    \item ${f}$ is linear.
\end{enumerate}
So, a non-coherent linear network channel under the rank weight is a linear generalized network channel. Furthermore,
\begin{enumerate}
    \setcounter{enumi}{9}
    \item $\mathcal{C}$ is a vector space by definition.
\end{enumerate}
Hence, a rank metric code over a non-coherent linear network channel is a linear generalized network code.

\subsection{Sum-rank Metric Code over a Non-coherent linear Network Channel}
The sum-rank metric code, first introduced in \cite{martinez2019reliable}, is a linear code equipped with a metric induced by the so-called \emph{sum-rank weight}. To define the sum-rank weight of a matrix $A$ with $n$ columns, $A$ is divided into $l$ blocks with each block containing $n_i$ columns for $i=1,2,\cdots,l$ such that $\sum_{i=1}^l n_i=n$. Without loss of generality, assume that  $A=(A_1|A_2|\cdots|A_l)$. The sum-rank weight of $A$, denoted as sr$(A)$, is the sum of the ranks (the sum-rank) of all the $l$ blocks of $A$, i.e., sr$(A)=\sum_{i=1}^l\mbox{rk}(A_i)$. Similar to the rank metric which uses the rank of the difference of two matrices to measure the distance between the two corresponding subspaces, the sum-rank metric uses the sum-rank of the difference of two matrices to measure the so-called \emph{sum-subspace distance}\cite{martinez2019reliable} between the two corresponding collections of subspaces.

In practice, a sum-rank metric code can be used in the so-called \emph{multishot network}, or $l$-shot network \cite{martinez2019reliable}. An $l$-shot network can be regarded as a random linear network being used $l$ times, where at each time, the network topology and the coding coefficients can be different and unknown to the sink nodes. So, the $l$-shot network is a non-coherent linear network channel. A sum-rank metric code over an $l$-shot non-coherent linear network channel can be cast into the framework of generalized network code as follows:
\begin{enumerate}[i)]
    \item The codeword set ${\mathcal{C}}$ is the space of $m\times k$ matrices over $\mathbb{F}_q$. Without loss of generality, we can assume that $\mathcal{C}=\mathbb{F}_q^{m\times k_1}\times\mathbb{F}_q^{m\times k_2}\times\cdots\times\mathbb{F}_q^{m\times k_l}=\mathbb{F}_q^{m\times k}$, with $k=k_1+k_2+\cdots+k_{l}$. The set of errors ${\mathcal{E}}$ is $\mathbb{F}_q^{m\times u_1}\times \mathbb{F}_q^{m\times u_2}\times \cdots\times \mathbb{F}_q^{m\times u_l}=\mathbb{F}_q^{m\times u}$, with $u=u_1+u_2+\cdots+u_l$ representing the length of the error vector. The received word set ${\mathcal{Y}}$ is $\mathbb{F}_q^{m\times n_1}\times \mathbb{F}_q^{m\times n_2}\times \cdots\times \mathbb{F}_q^{m\times n_l}=\mathbb{F}_q^{m\times n}$, with $n=n_1+n_2+\cdots+n_l$ representing the length of the received vector. Note that according to Definition 10 in \cite{martinez2019reliable}, in the case of the linearized Reed-Solomon code, which can be used as a sum-rank metric code, an additional constraint $l<q$ is imposed.

     \item The weight measure is the sum-rank weight sr. For all $z\in\mathbb{F}_q^{m\times u}$, we divide $z$ into $l$ blocks, namely, $z=(\tilde{z_1}|\tilde{z_2}|\cdots|\tilde{z_l})$ with $\tilde{z_i}\in\mathbb{F}_q^{m\times u_i}$. Then $\mbox{sr}(z)=\sum_{i=1}^{l}\mbox{rk}(\tilde{z_i})$, where rk$(\tilde{z_i})$ is the column rank weight of $\tilde{z_i}$.

     \item The transfer function is ${F}(x,z)={f}(x)+{h}(z)=x\cdot A+z\cdot B$, where $A=\mbox{diag}(A_1,A_2,\cdots,A_l)$ with $A_i\in\mathbb{F}_q^{k_i\times n_i}$ and $B=\mbox{diag}(B_1,B_2,\cdots,B_l)$ with $B_i\in\mathbb{F}_q^{u_i\times n_i}$. 
\end{enumerate} 

A sum-rank metric code can be used in an $l$-shot network, which is a random linear network being used for $l$ times. In the $i^{th}$ use of the network, a codeword $c_i\in\mathbb{F}_q^{m\times k_i}$ is transmitted in the network and an error vector $\tilde{z_i}\in\mathbb{F}_q^{m\times u_i}$ occurs in the network. The $i^{th}$ received word is $y_i\in\mathbb{F}_q^{m\times n_i}$. The network channel in the $i^{th}$ use can be represented as $F_i(x_i,\tilde{z_i})=x_iA_i+\tilde{z_i}B_i$, with $A_i$ being a $k_i\times n_i$ matrix and $B_i$ being a $u_i\times n_i$ matrix representing the transfer matrices for codewords and errors, respectively.
It is obvious that a sum-rank metric code with $l=1$ reduces to a rank-metric code.

Now we show that the non-coherent linear network channel associated with the sum-rank metric code is a linear generalized network channel by verifying Properties 1) to 10) in Section~\ref{sec7}:

\begin{enumerate}
     \item  $\mbox{sr}(z)\geq 0$ for all $z\in{\mathcal{E}}$, with equality if and only if $z$ is the zero matrix.
     \item ${\mathcal{E}}$ is $\mathbb{F}_q^{m\times u_1}\times \mathbb{F}_q^{m\times u_2}\times \cdots\times \mathbb{F}_q^{m\times u_l}=\mathbb{F}_q^{m\times u}$, which is a group with the binary operator `$\circ$' being matrix addition, and the zero error is the all-zero matrix $0$.
     \item ${\mathcal{Y}}$ is $\mathbb{F}_q^{m\times n_1}\times \mathbb{F}_q^{m\times n_2}\times \cdots\times \mathbb{F}_q^{m\times n_l}=\mathbb{F}_q^{m\times n}$, which is an abelian group with the commutative binary operator `$\oplus$' being matrix addition. The zero received word is the all-zero matrix 0.
     \item ${F}(x,z)={f}(x)+{h}(z)$, where $ {f}=x\cdot A$ and ${h}=z\cdot B$ are linear functions. Here, `$+$' denotes addition in the matrix space $\mathbb{F}_q^{m\times n_1}\times \mathbb{F}_q^{m\times n_2}\times \cdots\times \mathbb{F}_q^{m\times n_l}=\mathbb{F}_q^{m\times n}. $
     \item $\mbox{sr}(z+z')\leq \mbox{sr}(z)+\mbox{sr}(z')$, where the `$+$' in $z+z'$ is addition in the matrix space $\mathbb{F}_q^{m\times u_1}\times \mathbb{F}_q^{m\times u_2}\times \cdots\times \mathbb{F}_q^{m\times u_l}=\mathbb{F}_q^{m\times u}$ and the `$+$' in $\mbox{sr}(z)+\mbox{sr}(z')$ is addition in the integer group.
     \item  $\mbox{sr}(-z)=\mbox{sr}(z)$, with $-z$ being the additive inverse of $z$, i.e., $z+(-z)=0$.
     \item  We will prove in Appendix B that for every $z \in {\cal E}$ and $c_1, c_2 \ge 0$ such that $c_1+c_2 = \mbox{sr}(z)$, we can always find two matrices $z_1$ and $z_2$ such that $\mbox{sr}(z_1) = c_1$, $\mbox{sr}(z_2) = c_2$, and $z_1+z_2=z$.
 \end{enumerate}
 Therefore, a non-coherent linear network channel under the sum-rank weight is an error-linear generalized network channel. Moreover:
\begin{enumerate}
\setcounter{enumi}{7}
    \item By definition, ${\cal E}$, and ${\cal Y}$ are vector spaces.
    \item ${f}$ is linear.
\end{enumerate}
Therefore, a non-coherent linear network channel under the sum-rank weight is a linear generalized network channel. Furthermore,
\begin{enumerate}
    \setcounter{enumi}{9}
    \item $\mathcal{C}=\mathbb{F}_q^{m\times k}$ is a vector space.
\end{enumerate}
Hence, a sum-rank metric code over a non-coherent linear network channel is a linear generalized network code.

\section{ Error Correction And Detection Distances}\label{sec4}
 
\begin{definition}[Metric]
	A metric $D$ in a set $M$ is a function $D: M\times M\rightarrow \mathbb{R}$ which satisfies the following three axioms for all $x,y,z\in M$:
	\begin{itemize}
		\item Nonnegativity: $D(x,y)\geq 0$ with equality if and only if $x=y$.
		\item Symmetry: $D(x,y)=D(y,x)$.
		\item Triangular inequality: $D(x,y)+D(y,z)\geq D(x,z)$.
	\end{itemize}
\end{definition}

\vspace{2mm}
If all the axioms above except symmetry are satisfied, then $D $ is called a \emph{quasi-metric}\cite{quasi}. If all the axioms above except the triangular inequality are satisfied, then $D$ is called a \emph{semi-metric}\cite{semimetric}. In the rest of the paper, as long as $D$ satisfies nonnegativity, we will refer to it as a \emph{distance}.

In Section~\ref{sec7}, we defined the generalized network which is specified by a transfer function
${F}: {\cal C} \times {\cal E} \rightarrow {\cal Y}$, where  $\cal C$ is the set of codewords, $\cal E$ is the set of errors, and $\cal Y$ is the set of received words. Based on the weight measure $w$ with which the set $\cal E$ is equipped, we can define the {\em decoding ball}
 associated with a codeword.
 
 \vspace{2mm}
  \begin{definition}[Decoding Ball\cite{WeightProperty}]
 	\label{def1}
 	The \emph{decoding ball} of radius $c$ of codeword $x$ through a generalized network channel $F$ with respect to the weight measure $w$, denoted as $\Phi^F_w(x,c)$, is defined as\\
 	\begin{equation*}     
 		\Phi^F_w(x,c)=\{F(x,z):z\in \mathcal{E},\ w(z)\leq c\}.
 	\end{equation*}    
 \end{definition}

 \vspace{2mm} 
%
  It is obvious that $\Phi_w^F(x,c_1)\subseteq\Phi_w^F(x,c_2)$ if $c_1\leq c_2$. Also, by applying Property~1) in Section \ref{sec7}, we obtain
  \begin{equation}
  \Phi_w^F(x,0)=\{F(x,0)\}.
  \label{relwtjhlo}
  \end{equation}
  
 A generalized network code $({\cal C}, F)$ can be used for different error control functions. In this section, we discuss two of these functions, {\em error correction} and {\em error detection}. In the next few paragraphs, we will explain how a generalized network code $({\cal C}, F)$, together with a suitably designed decoder, can achieve error correction and error detection. 
 
 We first define the {\em minimum weight decoder} with respect to $w$, denoted by ${\rm MWD}_w$, which decodes a received word $y \in {\cal Y}$ as follows: 
first, find all the solutions of equation
\begin{equation}
F(x,z) = y
\label{1t98hqer}
\end{equation}
with $x \in {\cal C}$ and $z \in {\cal E}$ as variables. A pair $(x, z) \in {\cal C} \times {\cal E}$ is said to be a solution if it satisfies \eqref{1t98hqer}, and furthermore a minimum weight solution if $w(z)$ achieves the minimum among all the solutions. If for all minimum weight solutions $(x,z)$, we have $x = x^*$, then the decoder outputs $x^*$ as the decoded codeword. In this sense, we say that the decoder performs error correction on the received word $y$.
Otherwise, the decoder declares that an error is detected.

Next, we define another decoder related to ${\rm MWD}_w$ if the decoding balls $\{ \Phi_w^F(x,c) : x \in {\cal C} \}$ are disjoint, where $c$ is a nonnegative integer, called the {\em bounded distance}. For any received word $y \in {\cal Y}$, if $y \in \Phi(x,c)$ for a unique $x \in {\cal C}$, denoted by $x^*$, the decoder outputs $x^*$ as the decoded codeword. If $y$ is not in any of the decoding balls, the decoder declares that an error is detected. Such a decoder is called the minimum weight decoder with bounded distance $c$, denoted by ${\rm MWD}_w(c)$.

We see from \eqref{basicequation} and \eqref{relwtjhlo} that the decoding balls $\Phi_w^F(x,0)$, $x \in {\cal C}$ are disjoint. Therefore, the ${\rm MWD}_w(0)$ is always a valid decoder.

An error $z \in {\cal  E}$ is said to be {\em correctable} if for all $x \in {\cal C}$, if $x$ is transmitted through the generalized network channel $F$ and the error $z$ occurs, upon receiving $y = F(x,z)$, the ${\rm MWD}_w$ always outputs $x$ as the decoded codeword. In this sense, we say that the ${\rm MWD}_w$ can correct the error $z$. Otherwise, the error $z$ is said to be not correctable.

A nonzero error $z \in {\cal  E}$ is said to be {\em detectable} if for all $x \in {\cal C}$, if $x$ is transmitted through the generalized network channel $F$ and the error $z$ occurs, upon receiving $y = F(x,z)$, the ${\rm MWD}_w(0)$ always declares that an error is detected. In this sense, we say that the ${\rm MWD}_w(0)$ can detect the error $z$. Otherwise, the error $z$ is said to be not detectable.

If $x$ is transmitted through the generalized network channel $F$ and the zero error occurs, then $y = F(x,0) \in \Phi_w^F(x,0)$ by \eqref{relwtjhlo}. From the foregoing, $y \not\in \Phi_w^F(x',0)$ for any $x' \in {\cal C}$ such that $x' \ne x$. Therefore, the ${\rm MWD}_w(0)$ always output $x$ as the decoded codeword.

For a code $\cal C$, if all $z \in {\cal E}$ with $w(z) \le t$ are correctable, but there exists $z' \in {\cal E}$ with $w(z) = t+1$ such that $z'$ is not correctable, then we say that $\cal C$ can correct up to $t$ errors, or simply can correct $t$ errors. Likewise, if all $z \in {\cal E}$ with $w(z) \le t$ are detectable, but there exists $z' \in {\cal E}$ with $w(z') = t+1$ such that $z'$ is not detectable, then we say that $\cal C$ can detect up to $t$ errors, or simply can detect $t$ errors.

In this section, we consider the ${\rm MWD}_w(c)$ only for the special case $c=0$. The ${\rm MWD}_w(c)$ for $c > 0$ will be useful when we discuss joint error correction and detection in Section~\ref{sec5}.

\subsection{Two Distances for Error Correction And Detection}
 In Definition~\ref{def1}, the decoding  ball $\Phi^F_w(x,c)$ depends on both the generalized network channel $F$ and the weight measure $w$. In the rest of the paper, we will write $\Phi(x,c)$ instead of $\Phi^F_w(x,c)$ for simplicity when there is no ambiguity.

Based on the decoding ball, we now define two distances $D_0$ and $D_1$ for error correction and error detection, respectively. 

 \vspace{2mm}
\begin{definition}[Error Correction Distance\cite{WeightProperty}]
\label{def2}
 The \emph{error correction distance} between two codewords $x_1,x_2\in \mathcal{C}$ is defined as
\begin{equation*}    
\begin{aligned}
    D_0(x_1,x_2)=\min\limits_{|c_1-c_2| \leq 1,\atop c_1,c_2\geq 0}\{c_1+c_2:\Phi(x_1,c_1)\cap\Phi(x_2,c_2)\neq\emptyset\}.
\end{aligned}  
\end{equation*}
\end{definition}

Obviously, symmetry is satisfied by $D_0$. Also, for all $x\in\mathcal{C}$, since $\Phi(x,0)=\{F(x,0)\}$, we have $\Phi(x,0)\cap\Phi(x,0)\neq\emptyset$ and hence $D_0(x,x) = 0$. On the other hand, by our assumption that for $x_1\neq x_2$, $F(x_1,0)\neq F(x_2,0)$ (cf. (\ref{basicequation})), we have $\Phi(x_1,0)\cap\Phi(x_2,0)=\emptyset$ and hence $D_0(x_1,x_2) > 0$. So, $D_0(x_1,x_2)=0$ if and only if $x_1=x_2$, i.e.,  nonnegativity is satisfied by $D_0$.\footnote{We will prove in Corollary \ref{linearcorollary} that if $F$ is ``linear'', then $D_0$ is a metric and hence satisfies the triangular inequality.} However, it is not evident that $D_0$ satisfies the triangular inequality when $F$ is nonlinear. Therefore, $D_0$ is not necessarily a metric in general.  \\

 \begin{definition}[Error Detection Distance\cite{WeightProperty}]\label{def3}
The \emph{error detection distance} between two codewords $x_1,x_2\in \mathcal{C}$ is defined as
    \[D_1(x_1,x_2)=\min\limits_{c\geq 0}\{c:\Phi(x_1,0)\cap\Phi(x_2,c)\neq\emptyset\}.\]
 \end{definition}
 \vspace{2mm}
 
 Obviously, nonnegativity is satisfied by $D_1$ by the same argument as for $D_0$. However, $D_1$ is not a metric since in general symmetry is not satisfied when the channel function $F(x,z)$ is nonlinear. 
 
Similar to the relation of the Hamming distance between two codewords and the minimum distance of a block code, we can define the \emph{minimum error correction distance} and the \emph{minimum error detection distance} of a generalized network code as
\begin{equation}\label{d0}
    d_0^{\min}=\min\limits_{ x_1,x_2\in \mathcal{C},x_1\neq x_2}D_0(x_1,x_2)
\end{equation}
and
\begin{equation}\label{d1}
    d_1^{\min}=\min\limits_{ x_1,x_2\in \mathcal{C},x_1\neq x_2}D_1(x_1,x_2).
\end{equation}
To avoid triviality, we assume that $|{\cal C}| \ge 2$, i.e., the codebook contains at least two codewords. Since $D_0(x_1,x_2), D_1(x_1,x_2) > 0$ whenever $x_1 \ne x_2$, we have $d_0^{\min}, d_1^{\min} \ge 1$.

As proved in \cite{WeightProperty}, a network code with minimum error correction distance $d_0^{\min}$ can correct no more than $\left \lfloor{\frac{d_0^{\min}-1}{2}}\right \rfloor $ errors by a BMD decoder. On the other hand, a network code with minimum error detection distance $d_1^{\min}$ can detect no more than $d_1^{\min}-1$ errors by a BMD decoder.

In the next theorem, we establish that if $F$ is error-linear, then $D_0$ and $D_1$ coincide. 

 \vspace{2mm}
\begin{theorem}\label{d0d1coincide}
    Let $({\cal C}, F)$ be a generalized network code. If $F$ is error-linear, then $D_0(x_1,x_2)=D_1(x_1,x_2)$ for all $x_1,x_2\in\mathcal{C}$.
\end{theorem}

 \vspace{2mm}

	\begin{IEEEproof}
    Consider a generalized network code $({\cal C}, F)$, where $F$ is error-linear. Then Properties 1) to 7) in Section \ref{sec7} are satisfied. Fix any $x_1,x_2\in\mathcal{C}$.
    First, we prove that $D_1(x_1,x_2)\geq D_0(x_1,x_2)$. Let $D_1(x_1,x_2)=c$. Then by Definition \ref{def3}, there exists $z\in\mathcal{E}$ with $w(z)=c$ such that 
    \[
    f(x_1)\oplus h(0) = F(x_1, 0) = F(x_2, z) = f(x_2)\oplus h(z) .
    \] 
    So,
    \begin{equation}\label{aduaigiu}
    f(x_1)\oplus 0
    = f(x_2) \oplus h(z). 
    \end{equation}
    Now let $c=c_1+c_2$, with $c_1,c_2$ being non-negative integers such that $|c_1-c_2|\leq 1.$
    From Property 7) in Section \ref{sec7}, for a  vector $z\in\mathcal{E}$ with $w(z)=c$, there always exist $z_1,z_2\in\mathcal{E},$ with $w(z_1)=c_1,w(z_2)=c_2$, such that $z_2\circ z_1=z.$
     So, it follows from (\ref{aduaigiu}) that 
     \[
     f(x_1) = f(x_2)\oplus h(z_2\circ z_1)=f(x_2)\oplus h(z_2)\oplus h(z_1) .
     \]
Then
     \[f(x_1)\oplus (h(z_1))^{-1}=f(x_2)\oplus h(z_2),\]
     so that 
     \[
  F(x_1,z_1^{-1}) =   f(x_1)\oplus h(z_1^{-1})  =  f (x_1)\oplus (h(z_1))^{-1}   =f(x_2)\oplus h(z_2)=F(x_2, z_2).
    \]
     So, 
     by Definition \ref{def2} and Property 6) in Section \ref{sec7}, we obtain that \[D_0(x_1,x_2)\leq w(z_1^{-1})+w(z_2)=w(z_1)+w(z_2)=c_1+c_2=c=D_1(x_1,x_2).\]

    Next, we prove that $D_0(x_1,x_2)\geq D_1(x_1,x_2)$. Let $D_0(x_1,x_2)=c$. Then there exists $z_1,z_2\in\mathcal{E}$ with $w(z_1)=c_1$, $w(z_2)=c_2$, $|c_1-c_2|\leq 1$, and $c_1+c_2=c$, such that 
    \[
    f(x_1)\oplus h(z_1)=F(x_1,z_1)=F(x_2,z_2)=f(x_2)\oplus h(z_2).
    \]
    Then 
    \[
    f(x_1)\oplus 0=f(x_2)\oplus h(z_2)\oplus h(z_1)^{-1}= f(x_2)\oplus h(z_2)\oplus h(z_1^{-1}) =f(x_2)\oplus h(z_2\circ z_1^{-1}), 
    \]
    so that
    \[f(x_1)\oplus h(0)=f(x_2)\oplus h(z_2\circ z_1^{-1}).\]
     Since 
     \[
     w(z_2\circ z_1^{-1})\leq w(z_2)+w(z_1^{-1})=w(z_1)+w(z_2)=c_1+c_2=c,
     \]
     there exists $z\in\mathcal{E}$ with $w(z)\leq c$ such that $ f(x_1)\oplus h(0)=f(x_2)\oplus h(z)$, i.e., $F(x_1,0)=F(x_2,z)$. Therefore, $ D_0(x_1,x_2)=c\geq D_1(x_1,x_2)$. Together with the fact that $D_1(x_1,x_2)\geq D_0(x_1,x_2)$, we conclude that $D_0(x_1,x_2)=D_1(x_1,x_2)$. Hence, the theorem is proved.
\end{IEEEproof}

 \vspace{2mm}
 
 	\begin{remark}
It can readily be obtained from Lemma~6 in \cite{WeightProperty} that $D_0 = D_1$ for a coherent linear network code under the Hamming weight. The same  result is proved in Theorem~\ref{d0d1coincide} but under a much more general setting:
\begin{enumerate}
	\item 
	it is only required that the generalized network channel $F$ is error-linear, which incorporates not only the coherent linear network code but also the rank metric code and the sum-rank metric code as special cases;
	\item 
	the weight measure $w$ can be very general and only needs to satisfy Properties~1), 5), 6), and 7) in Section~\ref{sec7}, which includes the Hamming weight as a special case.
\end{enumerate}
\end{remark}

 \vspace{2mm}

\begin{corollary}
For an error-linear generalized network channel, $d_0^{min} = d_1^{min}$.	
\end{corollary}

 \vspace{2mm}
If the generalized network channel $F$ is not error-linear, i.e., not all the Properties 1) to 7) in Section \ref{sec7} are satisfied, $D_0$ and $D_1$ do not always coincide,  and hence $d_0^{\min}=d_1^{\min}$ does not hold in general. As an example, a nonlinear network code is given in \cite{WeightProperty} with $d_0^{\min}=3$ and $d_1^{\min}=2$, implying that it can correct $\lfloor \frac{d_0^{\min}-1}{2} \rfloor =1$ error, and can detect $d_1^{\min}-1=1$ error.  This is not possible for block codes in classical channel coding and linear network codes because for such codes, the number of errors that can be corrected is at most half of the number of errors that can be detected.

\subsection{Relations between $d_0^{\min}$ and $d_1^{\min}$}
Since the minimum distances for error correction and detection are in general not the same, it naturally brings up the question of how $d_0^{\min}$ and $d_1^{\min}$ are related. First of all, it is evident from Definitions~\ref{def2} and~\ref{def3} and the definitions in \eqref{d0} and \eqref{d1} that $d_0^{\min} \ge d_1^{\min}$. An example that $d_0^{\min} > d_1^{\min}$ is the nonlinear network code given in \cite{WeightProperty}, where $d_0^{\min}=3$ and $d_1^{\min}=2$. On the other hand, it is of interest to obtain a lower bound on $d_1^{\min}$ in terms of $d_0^{\min}$.



\vspace{2mm}

Now we give a lower bound on $D_1$.\\
\begin{theorem}\label{D0D1bound}
    For any $x_1,x_2\in \mathcal{C}$ with $x_1 \ne x_2$,  
    \[
    D_1(x_1,x_2) \geq  \left \lfloor {\frac{D_0(x_1,x_2)}{2}} \right \rfloor +1 .
    \]
\end{theorem}

\begin{IEEEproof}
	Let $D_0(x_1,x_2)=c$. Since $x_1 \ne x_2$, we have $c > 0$. Let $c_1 = \lceil c/2 \rceil$ and $c_2 = \lfloor c/2 \rfloor$, so that 
	\[
	c = c_1 + c_2, \ |c_1-c_2| \le 1, \ c_1 \ge 1, \ \mbox{and} \ c_2 + 1 \ge c_1.
	\]
	Since $D_0(x_1,x_2)= c_1 + c_2$ and $c_1 \ge 1$, by Definition~\ref{def2}, we have
	\begin{equation}
	\Phi(x_1,c_1-1)\cap\Phi(x_2,c_2)= \emptyset ,
	\label{qthpar}
	\end{equation}
	because $c_1-1 \ge 0$, $c_2 \ge 0$, 
	\[
	|(c_1-1)-c_2| = |(c_1-c_2)-1| \le | |c_1-c_2|-1| \le 1,
	\]
	and
	\[ 
	(c_1-1) + c_2 = (c_1+c_2)-1 = c-1 < c = D_0(x_1,x_2).	
	\]
	Moreover, since $c_1-1 \ge 0$, by Definition~\ref{def1}, \eqref{qthpar} implies that 
	\[
	\Phi(x_1,0)\cap\Phi(x_2,c_2)= \emptyset .
	\]
	Hence, by Definition~\ref{def3},
	\[
	D_1(x_1,x_2) \ge c_2+ 1  = \left \lfloor {c \over 2} \right \rfloor + 1 = \left \lfloor {\frac{D_0(x_1,x_2)}{2}} \right \rfloor +1 .
	\]
	The theorem is proved.
\end{IEEEproof}

\vspace{2mm}
The above theorem gives a lower bound on $D_1$ in terms of $D_0$. We have attempted to prove a lower bound on $D_0$ in terms of $D_1$ but in vain.

\vspace{2mm}

\begin{corollary}
\label{d1d0relation}
    $d_1^{\min}\geq \left \lfloor{\frac{d_0^{\min}}{2}}\right \rfloor+1$.
\end{corollary}

\vspace{2mm}

For a linear generalized network code, $d_0^{\min} = d_1^{\min}$. Let $d^{\min} 
= d_0^{\min} = d_1^{\min}$. Then Theorem~\ref{d1d0relation} becomes
$d^{\min}\geq\left \lfloor{\frac{d^{\min}}{2}}\right \rfloor+1$. This inequality holds when $d^{\min}\geq 1$, which is the case with the assumption that $|C|\geq 2$.


\section{Distance for Joint Error Correction and Detection}\label{sec5}
The notions of error correction and error detection can be unified under the framework of joint error correction and detection introduced in \cite{WeightProperty}, and a characterization of joint error correction and detection for network coding was obtained therein.\\

 \begin{definition}[Joint error correction and detection]
 \label{jointerror}
     A generalized network code $(\mathcal{C},F)$ with weight measure $w$ is called $(c,c')$ \emph{joint error-correcting-and-detecting} if the ${\rm MWD}_w(c)$ can correct any error $z_1$ with $w(z_1)\leq c$ and detect any error $z_2$ with $c< w(z_2)\leq c+c'$.
 \end{definition}
 \vspace{2mm}


Note that for $(c,c')$ joint error correction, if we set $c'=0$, then it becomes the ordinary error correction. Likewise, if we set $c=0$, then it becomes the ordinary error detection. In the subsequent discussions, ``joint error correction and detection'' will be abbreviated to ``joint error correction'' for simplicity.
\vspace{2mm}
\begin{theorem}\cite{WeightProperty}
\label{joint}
    A generalized network code $(\mathcal{C},F)$ is $(c,c')$ joint error-correcting if and only if for all distinct $ x_1,x_2\in\mathcal{C},\ \Phi(x_1,c)\cap\Phi(x_2,c+c')=\emptyset$.
\end{theorem}
\vspace{2mm}


For the purpose of analyzing the joint error correction and detection capability of a generalized network code, we naturally introduce the following distance measure.\\

\begin{definition}[Distance for Joint Error Correction]
\label{jointdef}
    Consider a generalized network code $(\mathcal{C},F)$. The distance for joint error correction between two codewords $x_1,x_2\in\mathcal{C}$ is defined as
    \begin{equation*}
        D_2(x_1,x_2)=\min\limits_{ c_1,c_2\geq 0}\{c_1+c_2: \Phi(x_1,c_1)\cap\Phi(x_2,c_2)\neq\emptyset\}.
    \end{equation*}
\end{definition}

    Accordingly, we can define the minimum distance for joint error correction of a network code $\mathcal{C}$ as 
    \begin{equation}\label{d2}
        d_2^{\min}=\min\limits_{x_1,x_2\in\mathcal{C},x_1\neq x_2}D_2(x_1,x_2).
    \end{equation}

Note that $D_2$ is not a metric, as it does not always satisfy the triangular inequality when the channel function $F(x,z)$ is nonlinear. On the other hand, it is obvious that $D_2$ satisfies nonnegativity and symmetry, so it is a semi-metric.

With the definition of the minimum distance $d_2^{\min}$, we can readily obtain the following theorem regarding $(c,c')$ joint error correction.\\

\begin{theorem}
\label{d_2th}
    Let $c$ and $c'$ be nonnegative integers. A generalized network code $(\mathcal{C},F)$ is $(c,c')$ joint error-correcting if $d_2^{\min}\geq 2c+c'+1$.
\end{theorem}
\vspace{2mm}
\begin{IEEEproof}
    If $d_2^{\min}\geq 2c+c'+1$, then for all $\ x_1,x_2\in\mathcal{C},$ $D_2(x_1,x_2)\geq 2c+c'+1. $ So for all $x_1,x_2\in\mathcal{C},\Phi(x_1,c)\cap\Phi(x_2,c+c')=\emptyset$. Then according to Theorem \ref{joint}, $(\mathcal{C},F)$ is $(c,c')$ joint error-correcting.
    
\end{IEEEproof}
\vspace{2mm}

Note that in Theorem \ref{d_2th}, $d_2^{\min}\geq 2c+c'+1$ is not stated as a necessary condition for the generalized network code $(\mathcal{C},F)$ to be $(c,c')$ joint error-correcting. In fact, this condition is in general not necessary because $\Phi(x_1,c)\cap\Phi(x_2,c+c')=\emptyset$ is not equivalent to $d_2^{\min}\geq 2c+c'+1$. Specifically, if $\Phi(x_1,c)\cap\Phi(x_2,c+c')=\emptyset$ but $\Phi(x_1,c+1)\cap\Phi(x_2,c+c'-1)\neq\emptyset$, then $d_2^{\min}\leq 2c+c'$. We will define a refined distance later in this section so that we can obtain a necessary and sufficient condition for a generalized network code to be $(c,c')$ joint error-correcting.
\\

\subsection{Relations Between $d_0^{\min}$, $d_1^{\min}$ and $d_2^{\min}$}
In this subsection, we explore the relations between $d_0^{\min}$, $d_1^{\min}$ and $d_2^{\min}$.\\

\begin{theorem}
\label{d_2theo}
    Let $(\mathcal{C},F)$ be a generalized network code. For all $x_1,x_2\in\mathcal{C}$, we have $D_0(x_1,x_2)\geq D_2(x_1,x_2)$ and $D_1(x_1,x_2)\geq D_2(x_1,x_2).$ If $F$ is an error-linear generalized network channel, then $D_0(x_1,x_2)=D_1(x_1,x_2)=D_2(x_1,x_2)$.
\end{theorem}
\vspace{2mm}
\begin{IEEEproof}
First, we prove that $D_1(x_1,x_2)\geq D_2(x_1,x_2)$. Let $D_2(x_1,x_2)=c$. Then it follows from Definition \ref{jointdef} that $\Phi(x_1,0)\cap \Phi(x_2,c-1)=\emptyset,$ implying that $ D_1(x_1,x_2)> c-1.$ Therefore, $ D_1(x_1,x_2)\geq c=D_2(x_1,x_2)$.
    
    Next, we prove that $D_0(x_1,x_2)\geq D_2(x_1,x_2)$. Again, let $D_2(x_1,x_2)=c$. It follows from Definition \ref{jointdef} that $\Phi(x_1,c_1)\cap\Phi(x_2,c_2-1)=\emptyset$ for all $c_1,c_2\geq 0$ such that $c_1+c_2=c-1$. In particular, $\Phi(x_1,\left \lfloor{\frac{c}{2}}\right \rfloor)\cap\Phi(x_2,\left \lceil{\frac{c}{2}}\right \rceil-1)=\emptyset$. Since $|\left \lfloor{\frac{c}{2}}\right \rfloor-(\left \lceil{\frac{c}{2}}\right \rceil-1)|\leq 1,$ it implies $D_0(x_1,x_2)>c-1$ by Definition \ref{def2}. Hence, $ D_0(x_1,x_2)\geq c=D_2(x_1,x_2)$.

    Now, we prove that for an error-linear generalized network channel $F$, $D_1(x_1,x_2)=D_2(x_1,x_2)$. Since $F$ is error-linear, then Properties 1) to 7) in Section~\ref{sec7} are satisfied. As before, let $D_2(x_1,x_2)=c$. Then there exists $z_1,z_2\in\mathcal{E}$ with $w(z_1)=c_1,w(z_2)=c_2,$ and $ c_1+c_2=c$ such that \[   f(x_1)\oplus h(z_1)=f(x_2)\oplus h(z_2),\] or \[ f(x_1)\oplus 0=f(x_2)\oplus h(z_2)\oplus (h(z_1))^{-1}=f(x_2)\oplus(h(z_2)\oplus (h(z_1)^{-1})=f(x_2)\oplus h(z_2\circ (z_1)^{-1}).\] Since \[w(z_2\circ z_1^{-1})\leq w(z_2)+w(z_1^{-1})=w(z_1)+w(z_2)=c_1+c_2=c,\] there exists $z\in\mathcal{E}$ with $w(z)\leq c$ such that $ f(x_1)\oplus h(0)=f(x_2)\oplus h(z)$, which is $F(x_1,0)=F(x_2,z)$. Therefore, $ D_1(x_1,x_2)\leq c=D_2(x_1,x_2)$. Together with the fact that $D_1(x_1,x_2)\geq D_2(x_1,x_2)$ which is already proved at the beginning of the proof, we obtain that $D_1(x_1,x_2)=D_2(x_1,x_2)$. According to Theorem \ref{d0d1coincide}, we have $D_0(x_1,x_2)=D_1(x_1,x_2)$ when $F$ is error-linear. Hence, we conclude that $D_0(x_1,x_2)=D_1(x_1,x_2)=D_2(x_1,x_2)$.

\end{IEEEproof}

Note that in the proof of Theorem \ref{d_2theo}, $F$ is required to be error-linear instead of linear, i.e., only $h$ but not $f$ in the transfer function $F$ is required to be ``linear" (namely the group homomorphism of $h$); $f$ can be arbitrary and $\cal E$ and $\cal Y$ are not necessarily vector spaces. Therefore, Theorem \ref{d_2theo} says that for a generalized network channel $F$, if Properties 1) to 7) but not necessarily Properties 8 and 9 in Section~\ref{sec7} are satisfied, i.e., $F$ is error-linear but not necessarily linear, then the three distances $D_0, D_1$ and $D_2$ coincide.

\vspace{2mm}
The following corollary of Theorem \ref{d_2theo} is evident.
\vspace{2mm}
\begin{corollary}
\label{d_2co}
    For a generalized network code $(\mathcal{C},F)$, $d_0^{\min}\geq d_2^{\min}$ and $d_1^{\min}\geq d_2^{\min}$. When $F$ is error-linear, $d_0^{\min}= d_1^{\min} =d_2^{\min}$.
\end{corollary}
\vspace{2mm}

Corollary \ref{d_2co} asserts that $d_2^{\min}$ is a lower bound on $d_0^{\min}$. Next we prove an upper bound on $d_0^{\min}$ in terms of $d_2^{\min}$.
\vspace{2mm}

\begin{theorem}
\label{d2ged0}
    Let $(\mathcal{C},F)$ be a generalized network code. For all $x_1,x_2\in\mathcal{C}$ with $x_1\neq x_2$, we have $D_2(x_1,x_2)\geq \left \lceil{\frac{D_0(x_1,x_2)}{2}}\right \rceil$.
\end{theorem}
\vspace{2mm}
\begin{IEEEproof}
    Let $(\mathcal{C},F)$ be a generalized network code. Consider $x_1,x_2\in\mathcal{C}$ with $x_1\neq x_2$ and let $D_0(x_1,x_2)=c$. Let $c_1$ and $c_2$ be any nonnegative integers such that $\Phi(x_1,c_1)\cap\Phi(x_2,c_2)\neq\emptyset$. Then $\Phi(x_1,c_1+c_2)\cap\Phi(x_2,c_2+c_1)\neq\emptyset$. Since $D_0(x_1,x_2)=c$, we have $(c_1+c_2)+(c_2+c_1)\geq c$, which implies $c_1+c_2\geq \frac{c}{2}$. Since $c_1+c_2$ is an integer, we have $c_1+c_2\geq \left \lceil{\frac{c}{2}}\right \rceil=\left \lceil{\frac{D_0(x_1,x_2)}{2}}\right \rceil$. Thus we have shown that if
     \begin{equation}\label{zhongjian}
         \Phi(x_1,c_1)\cap\Phi(x_2,c_2)\neq\emptyset,
     \end{equation}
     then $c_1+c_2\geq \left \lceil{\frac{D_0(x_1,x_2)}{2}}\right \rceil$. By minimizing over all $c_1$ and $c_2$ satisfying (\ref{zhongjian}), we see that $D_2(x_1,x_2)\geq \left \lceil{\frac{D_0(x_1,x_2)}{2}}\right \rceil$.
     
     
    
\end{IEEEproof}

\vspace{4mm}

\begin{corollary}
\label{d2geod0co}
    For a network code $\mathcal{C}$ through network channel $F$, $d_2^{\min}\geq \left \lceil{\frac{d_0^{\min}}{2}}\right \rceil$.
\end{corollary}
\vspace{2mm}
\begin{example}
Let $(\mathcal{C},F)$ be a generalized network code. If $d_0^{\min}\geq 7$, we have $d_2^{\min}\geq \left \lceil{\frac{d_{0}^{min}}{2}}\right \rceil\geq 4$, implying that $C$ is guaranteed to be at least $(1,1)$ or $(0,3)$ joint error-correcting.
\end{example}
\vspace{4mm}

Theorem \ref{d1d0relation} gives a lower bound on $d_1^{\min}$ in terms of $d_0^{\min}$. On the other hand, from Corollaries \ref{d_2co} and \ref{d2geod0co}, we obtain  
\begin{equation*}
\begin{split}
    d_1^{\min}&\geq d_2^{\min}\geq \left \lceil{\frac{d_0^{\min}}{2}}\right \rceil
\end{split}
\end{equation*}
which is also a lower bound on $d_1^{\min}$ in terms of $d_0^{\min}$. However, this bound is less tight than the bound we obtain in Theorem \ref{d1d0relation} because $\left \lfloor{\frac{d_0^{\min}}{2}}\right \rfloor+1\geq \left \lceil{\frac{d_0^{\min}}{2}}\right \rceil$, where the $\geq$ is strict when $d_0^{\min}$ is odd. Also from Corollaries \ref{d_2co} and \ref{d2geod0co}, we obtain
\begin{equation*}
    d_0^{\min}\geq d_2^{\min}\geq \left \lceil{\frac{d_0^{\min}}{2}}\right \rceil,
\end{equation*}
giving a lower bound and an upper bound on $d_2^{\min}$ in terms of $d_0^{\min}$.  

As mentioned before, $D_1$ and $D_2$ are not metrics in general. The following theorem is a necessary condition for $D_1$ and $D_2$ to be metrics.
\vspace{2mm}

\begin{theorem}
\label{d1d2metric}
    Let $(\mathcal{C},F)$ be a generalized network code. If for all $x_1,x_2\in\mathcal{C}$, $D_1(x_1,x_2)=D_2(x_1,x_2)$, then $D_1$ and $D_2$ are metrics.
    
\end{theorem}
\vspace{2mm}
\begin{IEEEproof}
    Obviously, $D_2$ always satisfies nonnegativity and symmetry. So we need to prove that if $D_2(x_1,x_2)=D_1(x_1,x_2)$ for all $x_1,x_2\in\mathcal{C}$, then $D_2$ also satisfies the triangular inequality. It then follows that both $D_1$ and $D_2$ are metrics.

    Let $ x_1,x_2,x_3$ be codewords in $\mathcal{C}$ and assume that $D_1=D_2$ holds for all codeword pairs. Since $D_2$ is symmetrical, $D_1$ is also symmetrical. Let $D_1(x_1,x_2)=D_1(x_2,x_1)=c_1$. Since $D_1(x_2,x_1)=c_1$, we have $\Phi(x_1,c_1)\cap\Phi(x_2,0)\neq\emptyset$, implying that $F(x_2,0)\in\Phi(x_1,c_1)$. Similarly, let $D_1(x_2,x_3)=D_1(x_3,x_2)=c_2$. Since $D_1(x_2,x_3)=c_2$, we have $\Phi(x_3,c_2)\cap\Phi(x_2,0)\neq\emptyset$, implying that $F(x_2,0)\in\Phi(x_3,c_2)$. Since $F(x_2,0)\in\Phi(x_1,c_1)$ and $F(x_2,0)\in\Phi(x_3,c_2)$, we have $\Phi(x_1,c_1)\cap\Phi(x_3,c_2)\neq\emptyset$, implying that $D_2(x_1,x_3)\leq c_1+c_2=D_1(x_1,x_2)+D_1(x_2,x_3)=D_2(x_1,x_2)+D_2(x_2,x_3)$. So $D_2$ satisfies the triangular inequality. Hence, $D_2$ is a metric, and so is $D_1$.
    
    
\end{IEEEproof}

\vspace{3mm}

\begin{corollary}\label{linearcorollary}
    For an error-linear generalized network channel $F$, $D_0=D_1=D_2$, and they are all metrics.
\end{corollary}
\vspace{1mm}
\begin{IEEEproof}
   Consider a generalized network code $(\mathcal{C},F)$ where $F$ is error-linear. Combining Theorem \ref{d_2theo} and Theorem \ref{d1d2metric}, we obtain that $D_0=D_1=D_2$ and they are metrics.
\end{IEEEproof}

\subsection{Refined Distance For Joint Error Correction}
As discussed in Theorem \ref{d_2th}, $d_2^{\min}\geq 2c+c'+1$ is a sufficient but not necessary condition for a generalized network code $(\mathcal{C},F)$ to be $(c,c')$ joint error-correcting. In this subsection, we will present a refined version of $D_2$ and $d_2^{\min}$ from which we can obtain a necessary and sufficient condition for $(c,c')$ joint error correction.\\

\begin{definition}
\label{d2'def1}
    Let $(\mathcal{C},F)$ be a generalized network code and $x_1,x_2\in\mathcal{C}$. For a fixed nonnegative integer $c$, the refined distance for $(c,c')$ joint error correction between $x_1$ and $x_2$ is defined as
    \begin{equation*}
        D_2[c](x_1,x_2)=\min\limits_{c'\geq 0}\{c': \Phi(x_1,c)\cap\Phi(x_2,c+c')\neq\emptyset\}.
    \end{equation*}
\end{definition}

Note that $D_2[c]$ is the number of additional errors that can be detected by the ${\rm MWD}_w(c)$ in joint error correction, given that $c$ errors can be corrected. From the definition, $D_2[c](x_1,x_2)=0$ if $c\geq \left \lfloor{\frac{D_0(x_1,x_2)+1}{2}}\right \rfloor=\left \lfloor{\frac{D_0(x_1,x_2)-1}{2}}\right \rfloor+1$. Since for $c\geq \left \lfloor{\frac{D_0(x_1,x_2)+1}{2}}\right \rfloor$, we always have $\Phi(x_1,c)\cap\Phi(x_2,c)\neq\emptyset$, so that the minimum $c'$ such that $\Phi(x_1,c)\cap\Phi(x_2,c+c')\neq\emptyset$ is $0$. Also note that $D_2[0](x_1,x_2)=D_1(x_1,x_2)$ from Definition \ref{def3} and Definition \ref{d2'def1}.
\vspace{2mm}

\begin{proposition}\label{prop5.13}
    Let $(\mathcal{C},F)$ be a generalized network code. For all $x_1,x_2\in\mathcal{C}$, $D_2[c](x_1,x_2)$ is nonincreasing in $c$.
\end{proposition}
\vspace{2mm}
\begin{IEEEproof}
    Let $(\mathcal{C},F)$ be a generalized network code and $x_1,x_2\in\mathcal{C}$. If $x_1=x_2$, then $D_2[c](x_1,x_2)=0$ for all $c\geq0$ and the proposition is immediate. So, we only need to prove the case that $x_1\neq x_2.$ Let $\alpha,\beta$ be two integers such that $\alpha>\beta\geq0,$ and let $D_2[\alpha](x_1,x_2)=\alpha'$, and $D_2[\beta](x_1,x_2)=\beta'$. Note that since $\alpha,\beta$ are integers, we have $\alpha\geq \beta+1.$ We now prove the theorem by contradiction. Assume that $\alpha'>\beta'$, which is $\alpha'\geq\beta'+1$. By Definition \ref{d2'def1}, we obtain
    \begin{align}
         \Phi(x_1,\alpha)\cap\Phi(x_2,\alpha+\alpha'-1)&=\emptyset\label{3}\\
         \Phi(x_1,\beta)\cap\Phi(x_2,\beta+\beta')&\neq\emptyset.\label{4}
     \end{align}
     Note that (\ref{3}) comes from the fact that $D_2[\alpha](x_1,x_2)=~\alpha'$ and (\ref{4}) comes from $D_2[\beta](x_1,x_2)=\beta'.$ Since $\alpha\geq \beta+1$ and $\alpha'>\beta'$, we obtain \[\alpha+\alpha'-1\geq \beta+\alpha'>\beta+\beta',\]
     then 
     \begin{align*}
         \Phi(x_1,\beta)&\subseteq\Phi(x_1,\alpha)\\
         \Phi(x_2,\beta+\beta')&\subseteq\Phi(x_2,\alpha+\alpha'-1).
     \end{align*}
     So, (\ref{3}) and (\ref{4}) can not hold at the same time, i.e., we get a contradiction between (\ref{3}) and (\ref{4}). Hence, the proposition is proved.
\end{IEEEproof}
\vspace{2mm}

Write \[\tau(x_1,x_2)=\left \lfloor{\frac{D_0(x_1,x_2)+1}{2}}\right \rfloor\] and \[c^*(x_1,x_2)=\left \lfloor{\frac{D_0(x_1,x_2)}{2}}\right \rfloor.\]
Note that $\tau(x_1,x_2)=\tau(x_2,x_1)$ and $c^*(x_1,x_2)=c^*(x_2,x_1)$ by the symmetry of $D_0$. In the sequel, when there is no ambiguity, we will abbreviate $\tau(x_1,x_2)$ and $c^*(x_1,x_2)$ to $\tau$ and $c^*$, respectively.
\vspace{2mm}
\begin{theorem}\label{foundation}
    Let $(\mathcal{C},F)$ be a generalized network code. For all $x_1,x_2\in\mathcal{C}$, we have $D_2[c](x_1,x_2)\leq D_1(x_1,x_2)$, with equality when $c=0$. Moreover, $D_2[c](x_1,x_2)=0$ if and only if $c\geq \tau$.
\end{theorem}
\vspace{2mm}
\begin{IEEEproof}
    Let $(\mathcal{C},F)$ be a generalized network code and $x_1,x_2\in\mathcal{C}$. First, we prove that $D_2[c](x_1,x_2)\leq D_1(x_1,x_2)$, with equality when $c=0$. Let $D_1(x_1,x_2)=c_1$ and $D_2[c](x_1,x_2)=c_2$. According to the definitions of $D_1$ and $D_2[c]$, we have 
    \begin{equation}
    \begin{split}
        \Phi(x_1,0)&\cap\Phi(x_2,c_1)\neq\emptyset\\
        \Phi(x_1,0)&\cap\Phi(x_2,c_1-1)=\emptyset\\
        \Phi(x_1,c)&\cap\Phi(x_2,c+c_2)\neq\emptyset\\
        \Phi(x_1,c)&\cap\Phi(x_2,c+c_2-1)=\emptyset.
    \end{split}
    \end{equation}
    Since $\Phi(x_1,0)\subset \Phi(x_1,c)$, we have 
    \[\Phi(x_1,0)\cap\Phi(x_2,c+c_2-1)=\emptyset.\]
    Together with $\Phi(x_1,0)\cap\Phi(x_2,c_1)\neq\emptyset$, we see that $c_1> c+c_2-1$. Hence, $c_1\geq c_2$, or $D_1(x_1,x_2)\geq D_2[c](x_1,x_2)$. When $c=0$, it follows directly from the definitions that $D_1(x_1,x_2)=D_2[0](x_1,x_2)$.

    Next, we prove that $D_2[c](x_1,x_2)=0$ if and only if $c\geq \tau$.
    As discussed in the paragraph following Definition \ref{d2'def1}, if $c\geq \tau$, we have $D_2[c](x_1,x_2)=0$. So, we only need to prove that if $D_2[c](x_1,x_2)=0$, then $c\geq \tau$.

     Assume that $D_2[c](x_1,x_2)=0$. Then $\Phi(x_1,c)\cap\Phi(x_2,c)\neq \emptyset$, implying that $D_0(x_1,x_2)\leq 2c$. Since $c$ is a nonnegative integer, we have $c\geq \left \lceil{\frac{D_0(x_1,x_2)}{2}}\right \rceil= \left \lfloor{\frac{D_0(x_1,x_2)+1}{2}}\right \rfloor=\tau$. Hence $D_2[c](x_1,x_2)=0$ if and only if $c\geq \tau$.
\end{IEEEproof}
\vspace{4mm}

Similar to the definition of $d_2^{\min}$, for a network code $\mathcal{C}$ and a fixed nonnegative integer $c$, we can define the refined minimum distance for $(c,c')$ joint error correction as

\begin{equation*}
    d_2^{\min}[c]=\min\limits_{x_1,x_2\in\mathcal{C},x_1\neq x_2}D_2[c](x_1,x_2).    
\end{equation*}
Now we show that $d_2^{\min}[c]\geq c'+1$ is a necessary and sufficient condition for $\mathcal{C}$ to be $(c,c')$ joint error-correcting. \\

\begin{theorem}
\label{d2'theo1}
    A generalized network code $(\mathcal{C},F)$ is $(c,c')$ joint error-correcting if and only if $d_2^{\min}[c]\geq c'+1$.
\end{theorem}
\vspace{2mm}
\begin{IEEEproof}
    First, we prove the ``only if" part. 
    According to Theorem~\ref{joint}, $(\mathcal{C},F)$ is $(c,c')$ joint error-correcting if and only if for all $ x_1,x_2\in\mathcal{C}$, $\Phi(x_1,c)\cap\Phi(x_2,c+c')=\emptyset$. So for all $x_1,x_2\in\mathcal{C}$, we have $D_2[c](x_1,x_2)\geq c'+1$ according to Definition \ref{d2'def1}. Hence, $d_2^{\min}[c]\geq c'+1$.\\
    
    Next, we prove the ``if" part. 
    If $d_2^{\min}[c]\geq c'+1$, then for all $x_1,x_2\in\mathcal{C}$, we have $D_2[c](x_1,x_2)\geq c'+1$. In that case, for all $x_1,x_2\in\mathcal{C}$, $\Phi(x_1,c)\cap\Phi(x_2,c+c')=\emptyset$, so $\mathcal{C}$ is $(c,c')$ joint error-correcting according to Theorem \ref{joint}.
    
\end{IEEEproof}

\vspace{2mm}



\begin{lemma}
\label{lemma1}
    Consider a generalized network code $(\mathcal{C},F)$ and any $x_1,x_2\in\mathcal{C}$. Then
    \begin{itemize}
        \item  $D_2[c^*](x_1,x_2)=D_2[c^*](x_2,x_1)=0$ if $D_0(x_1,x_2)$ is even; 
        \vspace{2mm}
        \item  $\min\left(D_2[c^*](x_1,x_2),\ D_2[c^*](x_2,x_1)\right)=1$ if\\
        $D_0(x_1,x_2)$ is odd. 

    \end{itemize}
\end{lemma}
\vspace{2mm}
\begin{IEEEproof}
     When $D_0(x_1,x_2)$ is even, we have \[c^*=D_0(x_1,x_2)/2.\] We see from Definition \ref{def2} that
\begin{align}
    \Phi(x_1,c^*)\cap\Phi(x_2,c^*)\neq\emptyset.
\end{align}
Then the minimum $c'$ such that \[\Phi(x_1,c^*)\cap\Phi(x_2,c^*+c')\neq\emptyset\] is $0$, and 
the minimum $c''$ such that \[\Phi(x_2,c^*)\cap\Phi(x_1,c^*+c'')\neq\emptyset\] is also $0$. Hence, by Definition \ref{d2'def1}, $D_2[c^*](x_1,x_2)=D_2[c^*](x_2,x_1)=0$. 
On the other hand, when $D_0(x_1,x_2)$ is odd, we have $D_0(x_1,x_2)=2c^*+1$. Then from Definition \ref{def2},  either
\begin{align*}
        \Phi(x_1,c^*)\cap\Phi(x_2,c^*+1)\neq \emptyset\\
         \Phi(x_1,c^*)\cap\Phi(x_2,c^*)= \emptyset
    \end{align*}
    or
    \begin{align*}
        \Phi(x_1,c^*+1)\cap\Phi(x_2,c^*)\neq \emptyset\\
         \Phi(x_1,c^*)\cap\Phi(x_2,c^*)= \emptyset.
    \end{align*}
Therefore, according to Definition \ref{d2'def1}, either
    \begin{align*}
        D_2[c^*](x_1,x_2)=1,
       \quad D_2[c^*](x_2,x_1)\geq1
    \end{align*}
    or
    \begin{align*}
        D_2[c^*](x_2,x_1)=1,
        \quad D_2[c^*](x_1,x_2)\geq 1.
    \end{align*}
    Hence, $\min\left(D_2[c^*](x_1,x_2),\ D_2[c^*](x_2,x_1)\right)=1$.
\end{IEEEproof}
\vspace{3mm}

\begin{theorem}
\label{d0d1d2'}
    Consider a generalized network code $(\mathcal{C},F)$ and any $x_1,x_2\in\mathcal{C}$. Then
    \[2c^*+ \min\left(D_2[c^*](x_1,x_2),D_2[c^*](x_2,x_1)\right)=D_0(x_1,x_2).\] 
\end{theorem}
\vspace{2mm}
\begin{IEEEproof}
    First consider the case that $D_0(x_1,x_2)$ is even, or $D_0(x_1,x_2)=2c^*.$ 
    According to Lemma \ref{lemma1}, we have \[D_2[c^*](x_1,x_2)=0.\] Since $D_0$ is symmetrical, we have \[c^*=\left \lfloor{\frac{D_0(x_1,x_2)}{2}}\right \rfloor=\left \lfloor{\frac{D_0(x_2,x_1)}{2}}\right \rfloor.\] By applying Lemma \ref{lemma1} with the roles of $x_1$ and $x_2$ exchanged, we obtain \[D_2[c^*](x_2,x_1)=0.\] Therefore, \[2c^*+ \min\left(D_2[c^*](x_1,x_2),D_2[c^*](x_2,x_1)\right)=D_0(x_1,x_2).\] 
    
     Now consider the case that $D_0(x_1,x_2)$ is odd, or \[D_0(x_1,x_2)=2c^*+1.\] Then according to Lemma \ref{lemma1}, 
    \begin{align*}
        2c^*+ &\min\left(D_2[c^*](x_1,x_2),D_2[c^*](x_2,x_1)\right)\\&=2c^*+1\\&=D_0(x_1,x_2).
    \end{align*}
    Combining the two cases for $D_0(x_1,x_2)$, we conclude that 
    \begin{align*}
        2c^*+ \min\left(D_2[c^*](x_1,x_2),D_2[c^*](x_2,x_1)\right)=D_0(x_1,x_2).
    \end{align*}

\end{IEEEproof}
\vspace{3mm}    


The next theorem gives the relation between $D_2$ and $D_2[c]$.\\

\begin{theorem}
\label{d2'coro1}
 Let $(\mathcal{C},F)$ be a generalized network code. For any $x_1, x_2\in\mathcal{C}$,
 \begin{align}\label{minc}
    \begin{split}
     D_2(x_1,x_2)=\min\Big(\min\limits_{  c }\{2c+D_2[c](x_1,x_2)\},\ \min\limits_{  c }\{2c+D_2[c](x_2,x_1)\}\Big).
     \end{split}
 \end{align}
\end{theorem}
\vspace{2mm}
\begin{IEEEproof}
First, we prove that 
\begin{align*}
D_2(x_1,x_2)\leq\min\Big(\min\limits_{  c }\{2c+D_2[c](x_2,x_1)\},\ \min\limits_{  c }\{2c+D_2[c](x_1,x_2)\}\Big).
\end{align*}
Consider
\begin{align*}
    D_2&(x_1,x_2)\\&=\min\limits_{ c_1,c_2\geq 0}\{c_1+c_2: \Phi(x_1,c_1)\cap\Phi(x_2,c_2)\neq\emptyset\}\\
    &\leq \min\limits_{ c_2\geq c_1\geq 0}\{c_1+c_2: \Phi(x_1,c_1)\cap\Phi(x_2,c_2)\neq\emptyset\}\\
    &=\min\limits_{ c\geq 0\atop c'\geq 0}\{2c+c': \Phi(x_1,c)\cap\Phi(x_2,c+c')\neq\emptyset\}\\
    &=\min\limits_{ c\geq 0}\min\limits_{ c'\geq 0}\{2c+c': \Phi(x_1,c)\cap\Phi(x_2,c+c')\neq\emptyset\}\\
    &\leq\min\limits_{ 0\leq c\leq \tau}\min\limits_{ c'\geq 0}\{2c+c': \Phi(x_1,c)\cap\Phi(x_2,c+c')\neq\emptyset\}\\
    &=\min\limits_{ 0\leq c\leq \tau}\left\{2c+\min\limits_{ c'\geq 0}\{c': \Phi(x_1,c)\cap\Phi(x_2,c+c')\neq\emptyset\}\right\}\\
    &=\min\limits_{ 0\leq c\leq \tau}\left\{2c+D_2[c](x_1,x_2)\right\}.
\end{align*}
Since $D_2$ is symmetrical, by exactly the same argument, we have
\[D_2(x_1,x_2)= D_2(x_2,x_1)\leq \min\limits_{0\leq c\leq \tau}\{2c+D_2[c](x_2,x_1)\}\]
Thus we see that 
\begin{align*}
    D_2(x_1,x_2)\leq\min\Big(\min\limits_{ 0\leq c \leq\tau}\{2c+D_2[c](x_1,x_2)\},\
    \min\limits_{ 0\leq c \leq\tau}\{2c+D_2[c](x_1,x_2)\}\Big).
\end{align*}
Since $D_2[c](x_1,x_2)=0$ for $c\geq \tau$ (cf. Theorem \ref{foundation}), the $\min\limits_{0\leq c\leq \tau}$ above can be replaced by $\min\limits_c$. Therefore, 
\begin{align*}
    D_2(x_1,x_2)\leq\min\Big(\min\limits_{ c }\{2c+D_2[c](x_1,x_2)\},\ \min\limits_{  c }\{2c+D_2[c](x_1,x_2)\}\Big).
\end{align*}
Next we prove that 
\begin{align*}
    D_2(x_1,x_2)\geq\min\Big(\min\limits_{  c }\{2c+D_2[c](x_1,x_2)\},\  \min\limits_{  c }\{2c+D_2[c](x_1,x_2)\}\Big).
\end{align*}
Consider 
\begin{align*}
    D_2&(x_1,x_2)\\&=\min\limits_{ c_1,c_2\geq 0}\{c_1+c_2: \Phi(x_1,c_1)\cap\Phi(x_2,c_2)\neq\emptyset\}\\
    &\geq \min\limits_{  c_1,c_2\geq 0}\{c_1+c_2: (\Phi(x_1,c_1)\cap\Phi(x_2,c_2)\neq\emptyset) \atop \vee\ (\Phi(x_1,c_2)\cap\Phi(x_2,c_1)\neq\emptyset)\}\\
     &= \min\limits_{  c_1\geq c_2\geq 0}\{c_1+c_2: (\Phi(x_1,c_1)\cap\Phi(x_2,c_2)\neq\emptyset) \atop \vee\ (\Phi(x_1,c_2)\cap\Phi(x_2,c_1)\neq\emptyset)\}.\\
\end{align*}
The last equality above is justified because for the set to be minimized, the quantity to be minimized (i.e., $c_1+c_2$) and the predicate specifying the set are both symmetrical in $c_1$ and $c_2$. Then,
\begin{align*}
    D_2&(x_1,x_2)\\&\geq\min\limits_{  c_1\geq c_2\geq 0}\{c_1+c_2: (\Phi(x_1,c_1)\cap\Phi(x_2,c_2)\neq\emptyset) \atop \vee\ (\Phi(x_1,c_2)\cap\Phi(x_2,c_1)\neq\emptyset)\}\\
    &=\min\limits_{  c\geq 0\atop c'\geq 0}\{2c+c': (\Phi(x_1,c+c')\cap\Phi(x_2,c)\neq\emptyset) \atop \vee\ (\Phi(x_2,c+c')\cap\Phi(x_1,c)\neq\emptyset)\}\\
     &=\min\limits_{  c\geq 0}\left\{2c+\min\limits_{c'\geq 0}\{c': (\Phi(x_1,c+c')\cap\Phi(x_2,c)\neq\emptyset)\right.\\&\left.\quad\quad\ \  \vee\ (\Phi(x_2,c+c')\cap\Phi(x_1,c)\neq\emptyset)\}\frac{ }{ }\right\}\\
     &=\min\limits_{c\geq 0}\Big\{ 2c+\min(\min\limits_{c'\geq 0}\{c':\Phi(x_1,c+c')\cap\Phi(x_2,c)\neq\emptyset \},\\  &\quad\quad\ \min\limits_{c'\geq 0}\{c':\Phi(x_2,c+c')\cap\Phi(x_1,c)\neq\emptyset \})\Big\}\\
     &=\min\Big(\min\limits_{c\geq 0}\{2c+\min\limits_{c'\geq 0}\{c':\Phi(x_1,c+c')\cap\Phi(x_2,c)\neq\emptyset\}\},\\
     &\quad\quad\  \min\limits_{c\geq 0}\{2c+\min\limits_{c'\geq 0}\{c':\Phi(x_2,c+c')\cap\Phi(x_1,c)\neq\emptyset\}\}\Big)\\
     &= \min\Big(\min\limits_{  c\geq 0 }\{2c+D_2[c](x_1,x_2)\},\\&\ \quad\quad\ \min\limits_{   c\geq 0 }\{2c+D_2[c](x_2,x_1)\}\Big).
\end{align*}

Therefore,
\begin{align*}
    D_2(x_1,x_2)
    \geq\min\Big(\min\limits_{   c}\{2c+D_2[c](x_1,x_2)\},
    \  \min\limits_{    c }\{2c+D_2[c](x_2,x_1)\}\Big).
\end{align*}
Combining the two parts above, we see that  
\begin{align*}
    \begin{split}
    D_2(x_1,x_2)=\min\Big(\min\limits_{   c }\{2c+D_2[c](x_1,x_2)\},\  
    \min\limits_{   c }\{2c+D_2[c](x_2,x_1)\}\Big).
    \end{split}
\end{align*}   
    
\end{IEEEproof}
\vspace{3mm}

\begin{corollary}
\label{xincoro}
    Let $(\mathcal{C},F)$ be a generalized network code. Then \[d_2^{\min}=\min\limits_{c}\{2c+d_2^{\min}[c]\}.\] 
\end{corollary}
\vspace{2mm}
\begin{IEEEproof}
    This is obtained by minimizing the LHS and RHS of (\ref{minc}) over all $x_1,x_2\in\mathcal{C}.$ 
\end{IEEEproof}
\vspace{3mm}

\begin{theorem}
\label{d2'coro2}
    Let $(\mathcal{C},F)$ be a generalized network code where $F$ is error-linear. Then for all $ x_1,x_2\in\mathcal{C}$, $x_1\neq x_2$ and $0\leq c\leq \tau$,  \[2c+D_2[c](x_1,x_2)=D_2(x_1,x_2).\]
\end{theorem}

\begin{IEEEproof}
    As discussed in the proof of Theorem \ref{d2'coro1}, the $\min\limits_c$ in (\ref{minc}) can be replaced by $\min\limits_{0\leq c\leq \tau}$. Thus to prove the theorem, it suffices to show that $2c+D_2[c](x_1,x_2)=2c+D_2[c](x_2,x_1)$ and that it is a constant for all $0\leq c\leq \tau$. We will use the same notations as in the proof of Theorem \ref{d_2theo}.
    
    Let $c_1,c_2 \leq \tau$ be two fixed but arbitrary distinct nonnegative integers. Assume that 
        \begin{align}\label{7}
        \begin{split}
            D_2[c_1](x_1,x_2)&=c_1' \\
            D_2[c_2](x_1,x_2)&=c_2'.
        \end{split}
        \end{align}
     According to Definition \ref{d2'def1}, we have 
    \begin{align*}
        \Phi(x_1,c_1)\cap \Phi(x_2,c_1+c_1')&\neq \emptyset\\
     \Phi(x_1,c_2)\cap \Phi(x_2,c_2+c_2')&\neq \emptyset.
    \end{align*} 
    So, there exists $ z_1,z_1',z_2,z_2'\in\mathcal{E}$, with $w(z_1)=c_1$, $ w(z_1')=c_1+c_1',\ w(z_2)=c_2,$ and $ w(z_2')=c_2+c_2',$ such that $\ f(x_1)\oplus h(z_1)=f(x_2)\oplus h(z_1')$ and $\ f(x_1)\oplus h(z_2)=f(x_2) \oplus h(z_2').$ Then by the error-linearity of $F$, we have 
    \begin{align}\label{8}
    \begin{split}
        f(x_1)&=f(x_2)\oplus h(z_1'\circ z_1^{-1})\\
        f(x_1)&=f(x_2)\oplus h(z_2'\circ z_2^{-1}).
    \end{split}
    \end{align}
     Since $w$ satisfies Properties 5), 6) and 7) in Section \ref{sec7}, we have
     \begin{align*}
         c_1'&\leq w(z_1'\circ z_1^{-1})\leq w(z_1')+w(z_1^{-1})=w(z_1')+w(z_1)=2c_1+c_1'\\
         c_2'&\leq w(z_2'\circ z_2^{-1})\leq w(z_2')+w(z_2^{-1})=w(z_2')+w(z_2)=2c_2+c_2'.
     \end{align*}
    It follows from (\ref{8}) and Definition \ref{def3} that $D_1(x_1,x_2)\leq 2c_1+c_1'$ and $D_1(x_1,x_2)\leq 2c_2+c_2'$.

    On the other hand, from (\ref{7}) and Definition \ref{d2'def1}, for all ${z}_1,\tilde{z}_1',\tilde{z}_2,\tilde{z}_2'\in\mathcal{E}$ with $w(\tilde{z}_1)=c_1,\ w(\tilde{z}_1')=c_1+c_1'-1,\ w(\tilde{z}_2)=c_2,$ and $ w(\tilde{z}_2')=c_2+c_2'-1$,  
     we have \[\Phi(x_1,c_1)\cap\Phi(x_2,c_1+c_1'-1)=\emptyset\] and \[\Phi(x_1,c_2)\cap\Phi(x_2,c_2+c_2'-1)=\emptyset.\] This implies that \[f(x_1)\neq f(x_2)\oplus h(\tilde{z}_1'\circ\tilde{z}_1^{-1})\] and \[f(x_1)\neq f(x_2)\oplus h(\tilde{z}_2'\circ\tilde{z}_2^{-1}).\] Since 
     \begin{align*}
         c_1'-1&\leq w(\tilde{z}_1'\circ\tilde{z}_1^{-1})\leq 2c_1+c_1'-1\\
         c_2'-1&\leq w(\tilde{z}_2'\circ\tilde{z}_2^{-1})\leq 2c_2+c_2'-1,
     \end{align*}
      together with Definition \ref{def3}, we have $D_1(x_1,x_2)\geq 2c_1+c_1'$ and $D_1(x_1,x_2)\geq 2c_2+c_2'$.

    Combining the two parts above, we have $D_1(x_1,x_2)= 2c_1+c_1'$ and $D_1(x_1,x_2)= 2c_2+c_2'$, giving $2c_1+c_1'=2c_2+c_2'$. Hence, $2c+D_2[c](x_1,x_2)$ is a constant. On the other hand, using the same method, we can also prove that $2c+D_2[c](x_2,x_1)$ is a constant, which is equal to $D_1(x_2,x_1)$. From Corollary $1$ and Lemma $6$ in \cite{WeightProperty}, we know that $D_1(x_1,x_2)=D_0(x_1,x_2)$ and $D_1(x_2,x_1)=D_0(x_2,x_1)$ when $F$ is error-linear. Since $D_0$ is symmetrical, we obtain $D_1(x_1,x_2)=D_1(x_2,x_1)$, so that $2c+D_2[c](x_1,x_2)=2c+D_2[c](x_2,x_1)$, and they are constant for all $0\leq c \leq \tau$. Together with Theorem \ref{d2'coro1}, we see that $D_2(x_1,x_2)=2c+D_2[c](x_1,x_2)$ for all $0\leq c\leq \tau$.
    
\end{IEEEproof}
\vspace{2mm}

\begin{corollary}\label{coromini}
    Let $(\mathcal{C},F)$ be a generalized network code where $F$ is error-linear. Then for all $0\leq c\leq \tau$, we have $2c+d^{\min}_2[c]=d^{\min}_2$.
\end{corollary}
\vspace{3mm}

As proved in Theorem \ref{d_2theo}, if $F$ is an error-linear generalized network channel, then $D_0=D_1=D_2$. Next, we obtain a sufficient condition for $D_0=D_1=D_2$ when $\mathcal{C}$ is nonlinear. \\

\begin{theorem}
\label{uniform coro}
    Let $(\mathcal{C},F)$ be a generalized network code. If 
    \begin{align}\label{last1}
        D_2(x_1,x_2)=2c^*+D_2[c^*](x_1,x_2)=D_1(x_1,x_2)
    \end{align}
     for all $ x_1,x_2\in\mathcal{C}$ with $x_1\neq x_2$, then $D_0(x_1,x_2)=D_1(x_1,x_2)=D_2(x_1,x_2)$, and $D_0,D_1,D_2$ are all metrics.
\end{theorem}
\vspace{2mm}
\begin{IEEEproof}
    Assume that (\ref{last1}) holds for all $x_1,x_2\in\mathcal{C}$ such that $x_1\neq x_2$. Then 
    
     \begin{align*}
     \begin{split}
         D_2(x_1,x_2)&=2c^*+D_2[c^*](x_1,x_2)\\
        &\geq 2c^*+ \min(D_2[c^*](x_1,x_2),D_2[c^*](x_2,x_1)).
     \end{split}
     \end{align*}
     Since the RHS of the inequality above is equal to $D_0(x_1,x_2)$ by Theorem \ref{d0d1d2'}, we obtain 
     \begin{equation*}
         D_2(x_1,x_2)\geq D_0(x_1,x_2).
     \end{equation*}
    Together with Theorem \ref{d_2theo}, which says that \[D_2(x_1,x_2)\leq D_0(x_1,x_2),\] we obtain 
    \begin{equation}\label{last3}
        D_2(x_1,x_2)=D_0(x_1,x_2).
    \end{equation} 
    It then follows from (\ref{last1}) that $D_1(x_1,x_2)=D_2(x_1,x_2)=D_0(x_1,x_2)$. Together with Theorem \ref{d1d2metric}, we see that $D_0,D_1,D_2$ are all metrics. Hence, the theorem is proved.
\end{IEEEproof}
\vspace{3mm}

The sufficient condition in Theorem \ref{uniform coro} involves a relation among $D_1$, $D_2$ and $D_2[c]$. In Theorem \ref{function}, we will establish another sufficient condition for $D_0=D_1=D_2$ that involves only $D_2[c]$. We first prove the following lemma.
\vspace{2mm}
\begin{lemma}\label{d2'equ}
    Let $(\mathcal{C},F)$ be a generalized network code. For any $x_1,x_2\in \mathcal{C}$, the following two conditions are equivalent:
    \renewcommand{\theenumi}{S\arabic{enumi}}
     \begin{enumerate}
        \item \label{sta1} $D_2[c^*](x_1,x_2)\leq 1$ and $D_2[c^*](x_2,x_1)\leq 1$.
        \item \label{sta2} $D_2[c^*](x_1,x_2)=D_2[c^*](x_2,x_1)$.
     \end{enumerate}
\end{lemma}
\begin{IEEEproof}
     First we prove that \ref{sta1}) implies \ref{sta2}). 
     Obviously, $D_0(x_1,x_2)$ is either odd or even. Recall Lemma \ref{lemma1}. If $D_0(x_1,x_2)$ is odd, then \[\min(D_2[c^*](x_1,x_2),D_2[c^*](x_2,x_1))=1.\] Together with \ref{sta1}), we obtain that
     \begin{equation}\label{deng5}
         D_2[c^*](x_1,x_2)=D_2[c^*](x_2,x_1)=1.
     \end{equation}
     On the other hand, if $D_0(x_1,x_2)$ is even, we have 
     \begin{equation}\label{deng6}
         D_2[c^*](x_1,x_2)=D_2[c^*](x_2,x_1)=0.
     \end{equation}
     Combining (\ref{deng5}) and (\ref{deng6}), we have proved that \ref{sta1}) implies \ref{sta2}).
     

    Next, we prove that \ref{sta2}) implies \ref{sta1}). Obviously, by Lemma~\ref{lemma1}, if \ref{sta2}) is true and $D_0(x_1,x_2)$ is even, then \[D_2[c^*](x_1,x_2)=D_2[c^*](x_2,x_1)=0,\] and if $D_0(x_1,x_2)$ is odd, then \[D_2[c^*](x_1,x_2)=D_2[c^*](x_2,x_1)=1.\] In both cases, \ref{sta1}) is true. Hence, we have proved the lemma.
    
\end{IEEEproof}
\vspace{3mm}




\begin{theorem}\label{function}
    Let $(\mathcal{C},F)$ be a generalized network code. If for all $x_1,x_2\in\mathcal{C}$ with $x_1\neq x_2$, the conditions 
\renewcommand{\theenumi}{C\arabic{enumi}}
     \begin{enumerate}
        \item \label{con1} $D_2[c^*](x_1,x_2)\leq 1$
        \item \label{con2} ${g}(c)\triangleq 2c+D_2[c](x_1,x_2)$ for $0\leq c \leq \tau$ achieves its minimum at $c=0$ and $c=c^*$
     \end{enumerate}
     are both satisfied, then $D_0=D_1=D_2$ and they are all metrics.
\end{theorem}
\begin{IEEEproof} 
   Assume that \ref{con1}) and \ref{con2}) hold for all $x_1,x_2\in\mathcal{C}$ with $x_1\neq x_2.$ Since \ref{con1}) holds for all $x_1,x_2\in\mathcal{C}$, $x_1\neq x_2$, for a specific such pair $x_1$ and $x_2$, we have both $D_2[c^*](x_1,x_2)\leq 1$ and $D_2[c^*](x_2,x_1)\leq 1.$ Then by Lemma \ref{d2'equ}, we have \[D_2[c^*](x_1,x_2)=D_2[c^*](x_2,x_1).\]   Together with Theorem \ref{d0d1d2'}, we obtain 
   \begin{equation}\label{fangcheng}
       D_0(x_1,x_2)=2c^*+D_2[c^*](x_1,x_2).
   \end{equation}
   On the other hand, from \ref{con2}), we have 
   \begin{equation}\label{fangcheng2}
       g(c^*)=g(0)=\min\limits_c \{2c+D_2[c](x_1,x_2)\}.
   \end{equation}
   From Theorem \ref{foundation}, we have \[D_1(x_1,x_2)=D_2[0](x_1,x_2)=g(0).\] Together with (\ref{fangcheng2}), we obtain 
   \begin{equation}\label{fangcheng3}
       D_1(x_1,x_2)=g(0)=g(c^*)=2c^*+D_2[c^*](x_1,x_2).
   \end{equation}
   So, combining (\ref{fangcheng}), (\ref{fangcheng2}) and (\ref{fangcheng3}), we obtain
    \begin{equation}\label{deng1}
        D_1(x_1,x_2)=D_0(x_1,x_2)=\min\limits_c \{2c+D_2[c](x_1,x_2)\}
    \end{equation}
    and in the same way,
    \begin{equation}\label{deng2}
        D_1(x_2,x_1)=D_0(x_2,x_1)=\min\limits_c \{2c+D_2[c](x_2,x_1)\}.
    \end{equation}
    Since $D_0$ is symmetrical, from (\ref{deng1}) and (\ref{deng2}), we obtain 
\begin{align*}
    \min\limits_c \{2c+D_2[c](x_1,x_2)\}=\min\limits_c \{2c+D_2[c](x_2,x_1)\}.
\end{align*}
It then follows from Theorem \ref{d2'coro1} that
\begin{align}\label{deng3}
    D_2(x_1,x_2)=\min\limits_c \{2c+D_2[c](x_1,x_2)\}.
\end{align}
Combining (\ref{deng1}) and (\ref{deng3}), we obtain
    \begin{align*}
        D_0(x_1,x_2)=D_1(x_1,x_2)=D_2(x_1,x_2).
    \end{align*}
 Together with Theorem \ref{d1d2metric}, $D_0,D_1, D_2$ are all metrics. Hence, the theorem is proved.
\end{IEEEproof}
\vspace{3mm}

In Theorem \ref{uniform coro} and Theorem \ref{function}, we obtain two sufficient conditions for $D_0=D_1=D_2$. On the one hand, in Theorem \ref{uniform coro}, we consider a relation among $D_1, D_2$ and $D_2[c]$. On the other hand, in Theorem \ref{function}, we regard $2c+D_2[c]$ as a function of $c$, namely $g(c)$, and express $D_0, D_1, D_2$ in terms of $g(c)$. In the subsequent discussions, we will call the sufficient condition in Theorem \ref{uniform coro} the \emph{relation condition} and the sufficient condition in Theorem \ref{function} the \emph{function condition}.

Next, we prove that the function condition is no weaker than the relation condition.
\vspace{1mm}
\begin{corollary}\label{zuihou}
    Let $(\mathcal{C},F)$ be a generalized network code. The function condition is no weaker than the relation condition, i.e., for all $x_1,x_2\in\mathcal{C}$ with $x_1\neq x_2$, \ref{con1}) and \ref{con2}) imply (\ref{last1}).
\end{corollary}
\begin{IEEEproof}
    Close examination of the proof of Theorem~\ref{foundation} reveals that (\ref{fangcheng3}), (\ref{deng1}) and (\ref{deng3}) together implies that $D_2(x_1,x_2)=2c^*+D_2[c^*](x_1,x_2)=D_1(x_1,x_2),$ i.e., the relation condition. Hence, the function condition is no weaker than the relation condition.
    
    
    
\end{IEEEproof}
\vspace{2mm}

We have proved in Theorem \ref{d_2theo} that $D_0=D_1=D_2$ for an error-linear generalized network channel. So, it is natural to ask whether the relation condition and the function condition are satisfied through error-linear generalized network channels. 

Consider a generalized network code $(\mathcal{C},F)$ where $F$ is error-linear and $x_1,x_2\in\mathcal{C}$. By Theorem \ref{d2'coro2}, since $0\leq c^*\leq\tau$, it is obvious that the relation condition is satisfied, i.e., \[D_2(x_1,x_2)=2c^*+D_2[c^*](x_1,x_2)=D_2[0](x_1,x_2).\] On the other hand, by Theorem \ref{d2'coro2}, we obtain that $g(c)=2c+D_2[c](x_1,x_2)$ is a constant for $0\leq c\leq \tau$. So \ref{con2}) is satisfied. By (\ref{fangcheng3}) and (\ref{deng1}), we obtain \[2c^*+D_2[c^*](x_1,x_2)=D_0(x_1,x_2),\] together with the fact that $c^*=\left \lfloor{\frac{D_0(x_1,x_2)}{2}}\right \rfloor$, it is obvious that $D_2[c^*](x_1,x_2)\leq 1$. So \ref{con1}) is satisfied. Hence, if $F$ is error-linear, then both the relation condition and the function condition are satisfied.  

\section{A Toy Example}\label{sec6}
\begin{figure}
    \centering
    \captionsetup{justification=centering}
    \includegraphics[width=0.3\textwidth]{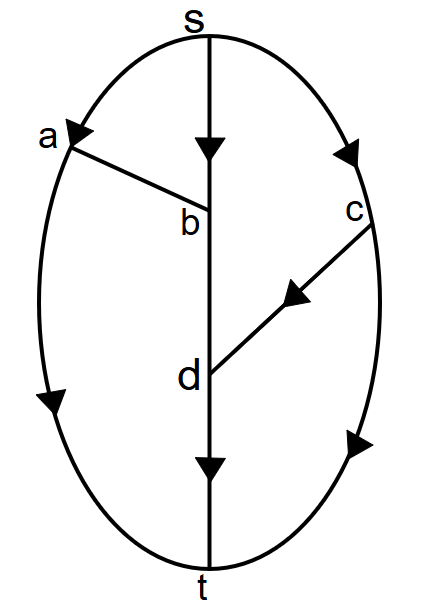}
    \caption{example network from \cite{WeightProperty}. $s$ is the source node and $t$ is the sink node.}
    
    \label{fig:enter-label}
\end{figure}
In this section, we give a toy example which was originally discussed in \cite{WeightProperty} to illustrate our results in the previous sections. The network is shown in Figure \ref{fig:enter-label}. As a generalized network code $(\mathcal{C},F)$, it is defined as the following:
\begin{enumerate}
    \item The codeset $\mathcal{C}$ is defined over $\mathbb{F}_3$ and $\mathcal{C}=\{(0,0,0),(1,1,1)\}$. 
    \item $F$ is defined by the network shown in Figure \ref{fig:enter-label}. Let $F_{(m,n)}$ be the value transmitted on channel (edge) $(m,n)$ with $m$ and $n$ being different nodes in Figure \ref{fig:enter-label}, and let $z_{(m,n)}$ be the error that occurs on channel $(m,n)$. Then $(F_{(s,a)},F_{(s,b)}, F_{(s,c)})$ is the codeword transmitted by the source node $s$, and $(F_{(a,t)},F_{(d,t)}, F_{(c,t)})$ is the vector received at the sink node~$t$.
\end{enumerate}
In this example, we have $F_{(s,a)}=F_{(a,b)}=F_{(a,t)}$ and $F_{(s,c)}=F_{(c,d)}=F_{(c,t)}$. $F_{(b,d)}$ is determined by $F_{(s,b)}$ and $F_{(a,b)}$ as follows\\

\begin{center}
\begin{tabular}{c|c|c|c|c|c|c|c|c|c}
    \toprule[1pt]
    
    $F_{(s,b)}$ & 0 & 0 & 0 & 1 & 1 & 1 & 2 & 2 & 2    \\ 
    \midrule[1pt]
    $F_{(a,b)}$ & 0 & 1 & 2 & 0 & 1 & 2 & 0 & 1 & 2     \\ 
    \midrule[2pt]
    $F_{(b,d)}$ & 0 & 0 & 0 & 2 & 1 & 0 & 0 & 0 & 0      \\ 
    \bottomrule[1pt]

\end{tabular}
\end{center}
\vspace{2mm}

and $F_{(d,t)}$ is determined by $F_{(b,d)}$ and $F_{(c,d)}$ as follows\\
\begin{center}
\begin{tabular}{c|c|c|c|c|c|c|c|c|c}
    \toprule[1pt]
    $F_{(b,d)}$ & 0 & 0 & 0 & 1 & 1 & 1 & 2 & 2 & 2    \\ \midrule[1pt]
    $F_{(c,d)}$ & 0 & 1 & 2 & 0 & 1 & 2 & 0 & 1 & 2     \\ \midrule[2pt]
    $F_{(d,t)}$ & 0 & 0 & 0 & 1 & 1 & 0 & 0 & 1 & 0      \\ \bottomrule[1pt]
\end{tabular}\\
\end{center}
\vspace{2mm}

Let $x_0=(0,0,0),\ x_1=(1,1,1)$. Then it is easy to obtain 
\begin{equation*}
\begin{split}
    \Phi(x_1,1)=\{(0,0,0),(1,0,0),(0,1,0),(0,0,1),(2,0,0),(0,2,0),(0,0,2)\}
\end{split}
\end{equation*}

 and
 
 \begin{equation*}
     \begin{split}
         \Phi(x_2,1)=\{(0,1,1),(1,0,1),
(1,1,0),(1,1,1),(1,0,2),(2,0,1),(2,1,1), (1,2,1),(1,1,2)\}.
     \end{split}
 \end{equation*} When two errors with value $2$ occur on $(s,a)$ and $(s,c)$, i.e., the error vector 
\[z=(z_{(s,a)},z_{(s,b)},z_{(s,c)},z_{(b,d)},z_{(c,d)},z_{(a,t)},z_{(b,t)},z_{(c,t)})\] is such that $z_{(s,a)}=z_{(s,c)}=2$ while the rest of the symbols are $0$, we have $F(x_0,0)=F(x_1,z)=(0,0,0)$. So,
\begin{align}
    \Phi(x_0,1)\cap \Phi(x_1,1) &=\emptyset.\label{64}\\
    \Phi(x_0,0)\cap \Phi(x_1,2)&\neq \emptyset \label{62}\\
    \Phi(x_0,1)\cap \Phi(x_1,2)&\neq \emptyset\label{63}
\end{align}
Note that (\ref{64}) implies
 \begin{align}
    \Phi(x_0,0)\cap \Phi(x_1,1)&=\emptyset\label{65}\\
    \Phi(x_0,1)\cap \Phi(x_1,0)&=\emptyset\label{66}
\end{align} because $\Phi(x_0,0)\subset \Phi(x_0,1)$ and $\Phi(x_1,0)\subset\Phi(x_1,1)$, respectively. According to the definitions of $d_0^{\min},\ d_1^{\min}$, and $d_2^{\min}$ (cf. (\ref{d0}), (\ref{d1}), and (\ref{d2}), respectively), (\ref{64}), and (\ref{63}) to (\ref{66}) imply that $d_0^{\min}=3$, while (\ref{62}), (\ref{65}), and (\ref{66}) imply $d_1^{\min}=2$ and $d_2^{\min}=2$. Thus, the code $\mathcal{C}$ can correct $\left \lfloor{\frac{d_0^{\min}-1}{2}}\right \rfloor=1$ error and can detect $d_1^{\min}-1=1$ error, where the latter is equivalent to $\mathcal{C}$ being $(0,1)$ joint error correcting.  Then 
\begin{itemize}
    \item $d_0^{\min}=3\geq 2=d_2^{\min}$ and $d_1^{\min}=2\geq 2=d_2^{\min}$, verifying Corollary \ref{d_2co};\vspace{1ex}
    \item $d_1^{\min}=2\geq 2=\left \lfloor{\frac{d_0^{\min}}{2}}\right \rfloor+1$, verifying Theorem \ref{d1d0relation};\vspace{1ex}
    \item $d_2^{\min}=2\geq 2= \left \lceil{\frac{d_{0}^{min}}{2}}\right \rceil$, verifying Corollary \ref{d2geod0co}.
\end{itemize} 


  Moreover, for $d_2^{\min}[c]$, (\ref{62}), (\ref{65}), and (\ref{66}) imply that $d_2^{\min}[0]=2$, while (\ref{64}) and (\ref{63}) imply that $d_2^{\min}[1]=1$. So, we have $2c+d_2^{\min}[c]=2$ for $c=0$, $2c+d_2^{\min}[c]=3$ for $c=1$ and $2c+d_2^{\min}[c]=4$ for $c=2$. Note that by Proposition \ref{prop5.13}, $d_2^{\min}[c]\leq d_2^{\min}[2]=0$ for all $c\geq 2$, implying that $d_2^{\min}[c]=0$ for all $c\geq 2.$ Therefore, it is not necessary to consider $2c+d_2^{\min}[c]$ for $c\geq 3$ for the purpose of calculating the minimum in Corollary \ref{xincoro}. Thus $d_2^{\min}=2=\min\{2,3,4\},$ verifying Corollary \ref{xincoro}.

\section{Conclusion}\label{sec9}


Yang {\em et al.} \cite{WeightProperty} discovered the surprising result that unlike a classical block channel code (linear or nonlinear) and a linear network code, a nonlinear network code in general requires two different minimum distances to characterize its error correction capability and its error detection capability. Inspired by the idea that the channel $F$ will affect the distances between the codewords, we established the scheme of generalized network channel and generalized network code. Then, we systematically defined the distances for error correction and error detection under the scheme of the generalized network code. We considered the joint error correction and detection in the generalized network code and obtained a complete characterization by introducing a distance and its refined version for this purpose. We have enhanced our understanding of the relation between various distances for error correction and detection in generalized network codes by proving some bounds on these distances.

The following are two interesting problems for further research: 
\begin{enumerate}
\item 
While all the distances discussed in this work coincide for linear network coding, what is the condition for them to coincide for nonlinear network coding?
\item 
What is the relation between these distances and the network topology?
\item How can the distance for erasure correction be defined? What are the properties of the erasure correction distance?

\end{enumerate}


%



\section*{Acknowledgment}

We would like to thank Prof. Shenghao Yang for his useful suggestions at the beginning of this research.

\begin{appendices}

\section{Proof of Property 7) in Section \ref{sec7} for the Rank Weight}
\label{appendixA}
Let $z\in\mathbb{F}_q^{m\times u}$ be an error vector with $\mbox{rk}(z)=c_1+c_2$, where $c_1,c_2$ are non-negative integers. Now we show that there exists $z_1,z_2\in\mathbb{F}_q^{m\times u}$ such that $\mbox{rk}(z_1)=c_1$, $\mbox{rk}(z_2)=c_2$, and $z_1+z_2=z.$
\vspace{2mm}

Assume that $z=(v_1,v_2,\cdots v_u)$, with $v_i\in\mathbb{F}_q^{m}$ for $i=1,2,\cdots,u$. Since $\mbox{rk}(z)=c_1+c_2$, there exist $c_1+c_2$ independent columns in $z$. Without loss of generality, we can assume that $v_1,v_2,\cdots,v_{c_1+c_2}$ are linearly independent. Then for $c_1+c_2+1\leq j\leq u$, we can write $v_j=f_j+g_j$, where $f_j\in\mbox{span}\{v_1,v_2,\cdots,v_{c_1}\}$, and $g_j\in\mbox{span}\{v_{c_1+1},v_{c_1+2},\cdots,v_{c_1+c_2}\}$. 
Accordingly, we can construct \[z_1=(v_1,v_2,\cdots,v_{c_1},\overbrace{0,0,\cdots,0}^{c_2},f_{c_1+c_2+1},f_{c_1+c_2+2},\cdots,f_{u}),\] and \[z_2=(\overbrace{0,0,\cdots,0}^{c_1},v_{c_1+1},v_{c_1+2},\cdots,v_{c_1+c_2},g_{c_1+c_2+1},g_{c_1+c_2+2},\cdots,g_{u}).\]  
Obviously, $\mbox{rk}(z_1)=c_1$, $\mbox{rk}(z_2)=c_2$, and $z_1+z_2=z$. Hence, property 7) in Section \ref{sec7} for the rank weight is proved.

\section{Proof of Property 7) in Section \ref{sec7} for the Sum-Rank Weight}
\label{appendixB}
Let $z=(\tilde{z_1}|\tilde{z_2}|\cdots|\tilde{z_l})\in\mathcal{E}$ be an error vector, $\tilde{z_i}\in\mathbb{F}_q^{m\times u_i}$ for $1\leq i\leq l$, and $\mbox{sr}(z)=c_1+c_2$, where $c_1,c_2$ are non-negative integers. Let $\mbox{rk}(\tilde{z_i})=r_i$. Then $\mbox{sr}(z)=\sum_{i=1}^lr_i= c_1+c_2$. Now we show that there exist $z_1,z_2\in\mathcal{E}$ such that $\mbox{sr}(z_1)=c_1$, $\mbox{sr}(z_2)=c_2$, and $z_1+z_2=z.$

For all $i=1,2,\cdots,l$, there exist $e_i,d_i\geq0$ such that $r_i=e_i+d_i$, $\sum_{i=1}^l e_i=c_1$ and $\sum_{i=1}^ld_i=c_2$. The result in Appendix A implies that
for all $i$,  there exist $\tilde{z}_1^i,\tilde{z}_2^i\in\mathbb{F}_q^{m\times u_i}$, such that $\mbox{rk}(\tilde{z}_1^i)=e_i$, $\mbox{rk}(\tilde{z}_2^i)=d_i$, and $\tilde{z_i}=\tilde{z}_1^i+\tilde{z}_2^i$.    

Let $z_1=(\tilde{z}_1^1|\tilde{z}_1^2|\cdots|\tilde{z}_1^l)$ and $z_2=(\tilde{z}_2^1|\tilde{z}_2^2|\cdots|\tilde{z}_2^l)$ so that $z_1+z_2=z$. Then $\mbox{sr}(z_1)=\sum_{i=1}^l\mbox{rk}(\tilde{z}_1^i)=\sum_{i=1}^l e_i=c_1$, and $\mbox{sr}(z_2)=\sum_{i=1}^l\mbox{rk}(\tilde{z}_2^i)=\sum_{i=1}^l d_i=c_2$. Hence, property 7) in Section \ref{sec7} for the sum-rank weight is proved.

\end{appendices}

\bibliographystyle{IEEEtran}
\bibliography{ref.bib}
\end{document}